\documentclass[final, 11pt]{article}

\def\showauthornotes{0}
\def\showkeys{0}
\def\showdraftbox{0}

\usepackage{xspace,xcolor,enumerate}
\usepackage{amsmath,amsfonts,amssymb}
\usepackage{color,graphicx}

\ifnum\showkeys=1
\usepackage[color]{showkeys}
\fi

\usepackage{dsfont}

\ifnum\showauthornotes=1
\newcommand{\Authornote}[2]{{\sf\small\color{red}{[#1: #2]}}}
\newcommand{\Authorcomment}[2]{{\sf \small\color{gray}{[#1: #2]}}}
\newcommand{\Authorfnote}[2]{\footnote{\color{red}{#1: #2}}}
\else
\newcommand{\Authornote}[2]{}
\newcommand{\Authorcomment}[2]{}
\newcommand{\Authorfnote}[2]{}
\fi

\ifnum\showdraftbox=1
\newcommand{\draftbox}{\begin{center}
  \fbox{%
    \begin{minipage}{2in}%
      \begin{center}%
        \begin{Large}%
          \textsc{Working Draft}%
        \end{Large}\\
        Please do not distribute%
      \end{center}%
    \end{minipage}%
  }%
\end{center}
\vspace{0.2cm}}
\else
\newcommand{\draftbox}{}
\fi

\newtheorem{theorem}{Theorem}[section]

\newtheorem{definition}[theorem]{Definition}
\newtheorem{lemma}[theorem]{Lemma}
\newtheorem{remark}[theorem]{Remark}

\newtheorem{corollary}[theorem]{Corollary}
\newtheorem{claim}[theorem]{Claim}

\def\FullBox{\hbox{\vrule width 6pt height 6pt depth 0pt}}

\def\qed{\ifmmode\qquad\FullBox\else{\unskip\nobreak\hfil
\penalty50\hskip1em\null\nobreak\hfil\FullBox
\parfillskip=0pt\finalhyphendemerits=0\endgraf}\fi}

\def\qedsketch{\ifmmode\Box\else{\unskip\nobreak\hfil
\penalty50\hskip1em\null\nobreak\hfil$\Box$
\parfillskip=0pt\finalhyphendemerits=0\endgraf}\fi}

\newenvironment{proof}{\begin{trivlist} \item {\bf Proof:~~}}
   {\qed\end{trivlist}}

\newenvironment{proofof}[1]{\begin{trivlist} \item {\bf Proof
#1:~~}}
  {\qed\end{trivlist}}

\def\to{\rightarrow}
\def\eps{\varepsilon}
\def\epsilon{\varepsilon}
\def\e{\epsilon}
\def\d{\delta}
\def\phi{\varphi}
\def\cal{\mathcal}

\def\ra{\rightarrow}
\def\implies{\Rightarrow}

\newcommand{\defeq}{\stackrel{\mathrm{def}}=}     
\providecommand{\cod}{\mathop{\sf cod}\nolimits}

\renewcommand{\bar}{\overline} 

\newcommand{\ie}{i.e.,\xspace}

\newcommand{\etal}{et al.\xspace}

\newcommand{\R}{{\mathbb R}}
\newcommand{\E}{{\mathbb E}}
\newcommand{\C}{{\mathbb C}}
\newcommand{\N}{{\mathbb{N}}}
\newcommand{\Z}{{\mathbb Z}}
\newcommand{\F}{{\mathbb F}}

\newcommand{\B}{\{0,1\}\xspace}
\newcommand{\pmone}{\{-1,1\}\xspace}

\usepackage{nicefrac}

\newcommand{\abs}[1]{\ensuremath{\left\lvert #1 \right\rvert}}
\newcommand{\smallabs}[1]{\ensuremath{\lvert #1 \rvert}}
\newcommand{\norm}[1]{\ensuremath{\left\lVert #1 \right\rVert}}

\newcommand{\ip}[1]{\left\langle #1 \right\rangle}

\newcommand{\Esymb}{\mathbb{E}}
\newcommand{\Psymb}{\mathbb{P}}

\DeclareMathOperator*{\ExpOp}{\Esymb}

\DeclareMathOperator*{\ProbOp}{\Psymb}
\renewcommand{\Pr}{\ProbOp}

\newcommand{\prob}[1]{\Pr\left[{#1}\right]}
\newcommand{\Prob}[2]{\Pr_{{#1}}\left[{#2}\right]}
\newcommand{\ex}[1]{\ExpOp\left[{#1}\right]}
\newcommand{\Ex}[2]{\ExpOp_{{#1}}\left[{#2}\right]}

\newfont{\inhead}{eufm10 scaled\magstep1}

\newcommand{\poly}{{\mathrm{poly}}}

\newcommand{\suchthat}{{\;\; : \;\;}}

\newcommand{\inparen}[1]{\left(#1\right)}             
\newcommand{\inbraces}[1]{\left\{#1\right\}}           
\newcommand{\insquare}[1]{\left[#1\right]}             

\newcommand{\Szemeredi}{Szemer\'edi\xspace}

\newcommand{\Holder}{H\"{o}lder}
\newcommand{\Plunnecke}{Pl\"unnecke}

\newcommand{\algofont}[1]{\fontshape{normal} \texttt{#1}}

\newcommand{\FindQuadratic}{\algofont{Find-Quadratic}\xspace}
\newcommand{\LinearDecomposition}{\algofont{Linear-Decomposition}\xspace}
\newcommand{\edgetest}{\algofont{Edge-Test}\xspace}

\newcommand{\bsgtest}{\algofont{BSG-Test}\xspace}

\newcommand{\bogolyubov}{\algofont{Bogolyubov}\xspace}

\newcommand{\FindQuadraticAverage}{\algofont{Find-QuadraticAverage}\xspace}

\newcommand{\ModelTest}{\algofont{Model-Test}\xspace}

\newcommand{\uthreenorm}[1]{\norm{#1}_{U^3}}
\newcommand{\snorm}[1]{\norm{#1}_{S}}
\newcommand{\comp}[1]{{#1}^*}
\newcommand{\fhat}{\hat{f}}
\newcommand{\fxhat}{\hat{f_x}}
\newcommand{\fyhat}{\hat{f_y}}
\newcommand{\fxyhat}{\widehat{f_{x+y}}}

\newcommand{\trunc}[1]{\ensuremath \algofont{Truncate}_{[-B,B]}\inparen{#1}}
\newcommand{\q}{\bar{q}}
\newcommand{\ortho}[1]{{#1}^{\perp}}

\title{%
  Quadratic Goldreich-Levin Theorems
}%

\author{%
Madhur Tulsiani\thanks{%
    Princeton University and IAS, Princeton, NJ. Work supported by NSF grant CCF-0832797.
  }%
\and Julia Wolf\thanks{%
   Centre de Math\'ematiques Laurent Schwartz, \'Ecole Polytechnique, 91128 Palaiseau, France.
}%
}

\date{\today}

\begin{document}

\sloppy

\maketitle

\draftbox
\setcounter{page}{1}
\thispagestyle{empty}
\begin{abstract}
Decomposition theorems in classical Fourier analysis enable us to express a bounded function in
terms of few linear phases with large Fourier coefficients plus a part that is pseudorandom with
respect to linear phases. 
The Goldreich-Levin algorithm \cite{GoldreichL89} can be viewed as an algorithmic analogue of 
such a decomposition as it gives a way to efficiently find the linear phases associated with large Fourier
coefficients.

In the study of ``quadratic Fourier analysis'', higher-degree analogues of such decompositions have
been developed in which the pseudorandomness property is stronger but the structured part
correspondingly weaker. For example, it has previously been shown that it is possible to express a
bounded function as a sum of a few quadratic phases plus a part that is small in the $U^3$ norm, defined by 
Gowers for the purpose of counting arithmetic progressions of length 4. We give a polynomial time algorithm for computing such a decomposition.

A key part of the algorithm is a local self-correction procedure for Reed-Muller codes of order 2
(over $\F_2^n$) for a function at distance $1/2-\e$ from a codeword. Given a function $f:\F_2^n \to
\{-1,1\}$ at fractional Hamming distance $1/2-\e$ from a quadratic phase (which is a codeword of Reed-Muller
code of order 2), we give an algorithm that runs in time polynomial in $n$ and finds a 
codeword at distance at most $1/2-\eta$ for $\eta = \eta(\e)$.
This is an algorithmic analogue
of Samorodnitsky's result \cite{Samorodnitsky07}, 
which gave a tester for the above problem. 
To our knowledge, it represents the first
instance of a correction procedure for any class of codes, beyond the list-decoding radius.

In the process, we give algorithmic versions of results from additive combinatorics used in Samorodnitsky's 
proof and a refined version of the inverse theorem for the Gowers $U^3$ norm over $\F_2^n$.
\end{abstract}
\thispagestyle{empty}

\newpage
\setcounter{page}{1}

\section{Introduction}
Higher-order Fourier analysis, which has its roots in Gowers's proof of Szemer\'edi's 
Theorem \cite{GSzT4}, has experienced a significant surge in the number of available tools as well
as applications in recent years, including perhaps most notably Green and Tao's proof 
that there are arbitrarily long arithmetic progressions in the primes.

Across a range of mathematical disciplines, classical Fourier analysis is often applied in 
form of a \emph{decomposition theorem}: one writes a bounded function $f$ as
\begin{equation}\label{decomp}
f=f_1+f_2,
\end{equation}
where $f_1$ is a structured part consisting of the frequencies with large amplitude, 
while $f_2$ consists of the remaining frequencies and resembles uniform, or random-looking, noise.
Over $\F_2^n$, the Fourier basis consists of functions of the form $(-1)^{\langle \alpha, x\rangle}$ for $\alpha
\in F_2^n$, which we shall refer to as \emph{linear phase functions}. The part $f_1$ is then a (weighted) sum of a
few linear phase functions.

From an algorithmic point of view, efficient techniques are available to 
compute the structured part $f_1$. The Goldreich-Levin \cite{GoldreichL89} theorem gives an
algorithm which computes, 
with high probability, the large Fourier coefficients of $f: \F_2^n \to \{-1,1\}$ 
in time polynomial in $n$. One way of viewing this theorem is precisely as
an algorithmic version of the decomposition theorem above, where $f_1$ is the part consisting
of large Fourier coefficients of a function and $f_2$ is random-looking with respect to any
test that can only detect large Fourier coefficients.

It was observed by Gowers (and previously by Furstenberg and Weiss in the context of ergodic theory) 
that the count of certain patterns is \emph{not} almost invariant under the addition of a noise term 
$f_2$ as defined above, and thus a decomposition such as (\ref{decomp}) is not sufficient in 
that context. In particular, for counting 4-term arithmetic progressions a more sensitive 
notion of uniformity is needed. This subtler notion of uniformity, called \emph{quadratic uniformity}, 
is expressed in terms of the $U^3$ norm, which was introduced by Gowers in \cite{GSzT4} and which we shall define below.

In certain situations we may therefore wish to decompose the function $f$ as above, but 
where the random-looking part is quadratically uniform, meaning $\|f_2\|_{U^3}$ is small. 
Naturally one needs to answer the question as to what replaces the \emph{structured part}, which in
(\ref{decomp})  was defined by a small number of linear characters.

This question belongs to the realm of what is now called \emph{quadratic Fourier analysis}.
Its central building block, largely contained in Gowers's proof of Szemer\'edi's theorem 
but refined by Green and Tao \cite{GrTu3} and Samorodnitsky \cite{Samorodnitsky07}, is the 
so-called \emph{inverse theorem for the $U^3$ norm}, which states, roughly speaking, that a function with 
large $U^3$ norm correlates with a \emph{quadratic phase function}, by which we mean a function of
the form $(-1)^{q}$ for a quadratic form $q: \F_2^n \to \F_2$.

The inverse theorem implies that the structured part $f_1$ has quadratic structure in the 
case where $f_2$ is small in $U^3$, and starting with \cite{GrML} a variety of such 
\emph{quadratic decomposition theorems} have come into existence: in one formulation \cite{GW1}, 
one can write $f$ as
\begin{equation}\label{quaddecomp}
f=\sum_i \lambda_i (-1)^{q_i} + f_2 +h,
\end{equation}
where the $q_i$ are quadratic forms, the $\lambda_i$ are real coefficients such that 
$\sum_i |\lambda_i|$ is bounded, $\|f_2\|_{U^3}$ is small and $h$ is a small $\ell_1$ 
error (that is negligible in all known applications.)

In analogy with the decomposition into Fourier characters, it is natural to think 
of the coefficients $\lambda_i$ as the \emph{quadratic Fourier coefficients} of $f$. 
As in the case of Fourier coefficients, there is a trade-off between the complexity 
of the structured part and the randomness of the uniform part. In the case of the 
quadratic decomposition above, the bound on the $\ell^1$ norm of the coefficients 
$\lambda_i$ depends inversely on the uniformity parameter $\|f_2\|_{U^3}$.
However, unlike the decomposition into Fourier characters, the decomposition in terms of quadratic
phases is not necessarily unique, as the quadratic phases do not form a basis for the space of functions on $\F_2^n$.

Quadratic decomposition theorems have found several number-theoretic applications, 
notably in a series of papers by Gowers and the second author \cite{GW1,GW2,GW4}, 
as well as \cite{Candela} and \cite{HatamiL11}.

However, all decomposition theorems of this type proved so far have been of a rather abstract nature. 
In particular, work by Trevisan, Vadhan and the first author \cite{TrevisanTV09} uses 
linear programming techniques and boosting, while Gowers and the second author \cite{GW1} gave a 
(non-constructive) existence proof using the Hahn-Banach theorem. The boosting proof is constructive
in a very weak sense (see Section \ref{sec:decompositions}) but is quite far from giving an
algorithm for computing the above decompositions. We give such an algorithm in this paper.

\vspace{-10 pt}
\paragraph{A computer science perspective.}
Algorithmic decomposition theorems, such as the weak regularity lemma of
Frieze and Kannan \cite{FriezeK99} which decomposes a matrix as a small sum of cut matrices,
have found numerous application in approximately solving constraint satisfaction problems. From the
point of view of theoretical computer science, a very natural question to ask is
if the simple description of a bounded function as a small list of quadratic phases can be computed 
efficiently. In this paper we give a probabilistic 
algorithm that performs this task, using a number of refinements of ingredients in the proof 
of the inverse theorem to make it more efficient, which will be detailed below.

\vspace{-10 pt}
\paragraph{Connections to Reed-Muller codes.}
A building block in proving the decomposition theorem is an algorithm for the following problem: 
given a function $f:\F_2^n \to \{-1,1\}$, which is at Hamming distance at most $1/2 - \e$ from an
unknown quadratic phase $(-1)^q$, find (efficiently) a quadratic phase $(-1)^{q'}$ which is at
distance at most $1/2 - \eta$ from $f$, for some $\eta = \eta(\e)$.

This naturally leads to a connection with Reed-Muller
codes since for Reed-Muller codes of order 2, the codewords are precisely the (truth-tables of) 
quadratic phases.

Note that the list decoding radius of Reed-Muller codes of order 2 is $1/4$ \cite{GopalanKZ08,
  Gopalan10}, which means that if the distance were less than $1/4$, we could find \emph{all} such $q$, 
and there would only be $\poly(n)$ many of them. The distance here is greater than $1/4$ 
and there might be exponentially many (in $n$) such functions $q$. However, 
the problem may still be tractable as we
are required to find only \emph{one} such $q$ (which might be at a slightly larger 
distance than $q'$).

The problem of \emph{testing} if there is such a $q$ was considered by Samorodnitsky 
\cite{Samorodnitsky07}. We show that in fact, the result can be turned into a 
\emph{local self corrector} for Reed-Muller codes at distance $(1/2 - \e)$. We are not 
aware of any class of codes for which such a self-correcting procedure is known, 
beyond the list-decoding radius.

\subsection{Overview of results and techniques}\label{sec:results}

We state below the basic decomposition theorem for quadratic phases, which is obtained by 
combining Theorems \ref{thm:decomposition-general} and \ref{thm:FindQuadratic} proved later.
The theorem is stated in terms of the $U^3$ norm, defined formally in Section
\ref{sec:preliminaries}.

\begin{theorem}\label{thm:decomposition-intro}
Let $\e, \delta > 0$, $n \in \N$ and $B > 1$. Then there exists $\eta = \exp((B/\e)^C)$ and a
randomized algorithm running in time $O(n^4 \log n \cdot \poly(\eta,\log(1/\delta)))$ which,
given any function $g: X \to [-1,1]$ as an oracle, outputs with probability at least $1-\delta$ a
decomposition into quadratic phases
\[ g~=~ c_1 (-1)^{q_1} + \ldots + c_k (-1)^{q_k} + e + f \] 
satisfying $k \leq 1/\eta^2$,   $\uthreenorm{f} \leq \epsilon$, $\norm{e}_1 \leq 1/2B$ and $|c_i| \leq \eta$
for all $i$.
\end{theorem}

Note that in \cite{GW2} the authors had to work much harder to obtain a bound on the number of terms in the decomposition, rather than just the $\ell^1$ norm of its coefficients. Our decomposition approach gives such a bound immediately and is equivalent from a quantitative point of view: we can bound the number of terms here by $1/\eta^2$, which is
exponential in $1/\e$. 

It is possible to further strengthen this theorem by combining the quadratic phases obtained into
only $\poly(1/\e)$  \emph{quadratic averages}. Roughly speaking, each quadratic average is a sum of
few quadratic phases, which differ only in their linear part. We describe this in detail in
Section \ref{sec:refinement}.

The key component of the above decomposition theorem is the following self-correction procedure for
Reed-Muller codes of order 2 (which are simply truth-tables of quadratic phase functions).
The correlation between two functions $f$ and $g$ is defined as $\ip{f,g} = \Ex{x \in \F_2^n}{f(x)g(x)}$.
\begin{theorem}\label{thm:FindQuadratic-Intro}
Given $\e, \delta > 0$, there exists $\eta = \exp(-1/\e^{C})$ and
a randomized algorithm \FindQuadratic running in time $O(n^4 \log n \cdot \poly(1/\e, 1/\eta,
\log(1/\delta)))$ which, 
given oracle access to a function $f: \F_2^n \to \pmone$, either outputs a quadratic form $q(x)$
or $\bot$. The algorithm satisfies the following guarantee.
\begin{itemize}
\item If $\uthreenorm{f} \geq \e$, then with probability at
  least $1-\delta$ it finds a quadratic form $q$ such that $\ip{f,(-1)^q} \geq \eta$.
\item The probability that the algorithm outputs a quadratic form $q$ with $\ip{f,(-1)^q} \leq \eta/2$ is at 
most $\delta$.
\end{itemize}
\end{theorem}

We remark that all the results contained here can be extended to $\F_p^n$ for any constant $p$. We
choose to present only the case of $\F_2^n$ for simplicity of notation.

Our results for computing the above decompositions comprise various components.

\vspace{-10 pt}
\paragraph{Constructive decomposition theorems.}
We prove the decomposition theorem using a procedure which, at every step, tests
if a certain function has correlation at least $1/2-\e$ with a quadratic phase.
Given an algorithm to \emph{find} such a quadratic phase, the procedure gives
a way to combine them to obtain a decomposition.

Previous decomposition theorems have also used such procedures 
\cite{FriezeK99, TrevisanTV09}. However, they required that the quadratic
phase found at each step have correlation $\eta = O(\e)$, if one exists with
correlation $\e$. In particular, they require the fact that if we scale $f$ to change
its $\ell_{\infty}$ norm, the quantities $\eta$ and $\e$ would scale the same way 
(this would not be true if, say, $\eta = \e^2$).

We need and prove a general decomposition theorem, which works even as $\eta$ degrades
arbitrarily in $1/\eps$. This requires a somewhat more sophisticated analysis and
the introduction of a third error term for which we bound the $\ell_1$ norm.
\paragraph{Algorithmic versions of theorems from additive combinatorics.}
Samorodnitsky's proof uses several results from additive combinatorics, which produce large
sets in $\F_2^n$ with certain useful additive properties. The proof of the inverse theorem
uses the description of these sets.
However, in our setting, we do not have time to look at the entire set since they
may be of size $\poly(\e) \cdot 2^n$, as in the case of the Balog-\Szemeredi-Gowers
theorem described later. We thus work by building efficient sampling
procedures or procedures for efficiently deciding membership in such sets, 
which require new algorithmic proofs. 
 
A subtlety arises when one tries to construct such a testing procedure. Since
the procedure runs in polynomial time, it often works by sampling and estimating
certain properties and the estimates may be erroneous. This leads to some noise in
the decision of any such an algorithm, resulting a noisy version of the set (actually a
distribution over sets). We get around this problem by proving a robust version
of the Balog-\Szemeredi-Gowers theorem, for which we can ``sandwich'' the output of
such a procedure between two sets with desirable properties. This technique may be
useful in other algorithmic applications.
\paragraph{Local inverse theorems and decompositions involving quadratic averages.}
Samorodnitsky's inverse theorem says that when a function $f$ has $U^3$ norm $\e$, then
one can find a quadratic phase $q$ which has correlation $\eta$ with $f$, for 
$\eta = \exp(-1/\e^C)$. A decomposition then requires $1/\eta^2$, that is exponentially many (in
$1/\e$), terms.

A somewhat stronger result was implicit in the work of Green and Tao \cite{GrTu3}. They showed that there exists a subspace of codimension
$\poly(1/\e)$ and on all of whose cosets $f$ correlates polynomially with a quadratic phase. Picking a particular coset and extending that quadratic phase to the whole space gives the previous theorem. 

It turns out that the different quadratic phases on each coset in fact have the same quadratic part
and differ only by a linear term. This was exploited in \cite{GW1} to obtain a decomposition
involving only polynomially many quadratic objects, so-called \emph{quadratic averages}, which are
described in more detail in Section \ref{sec:refinement}.

We remark that the results of Green and Tao \cite{GrTu3} do not directly extend to the case of characteristic 2 since division by 2 is used at one crucial point in the argument. We combine their ideas with those of Samorodnitsky to give an algorithmic version of a decomposition theorem involving quadratic averages.

\section{Preliminaries}
\label{sec:preliminaries}

Throughout the paper, we shall be using Latin letters such as $x$, $y$ or $z$ to denote elements of
$\F_2^n$, while Greek letters $\alpha$ and $\beta$ are used to denote members of the dual space
$\widehat{\F_2^n} \cong\F_2^n$. We shall use $\delta$ as our error parameter, while $\eps, \eta,
\gamma$ and $\rho$ are variously used to indicate correlation strength between a Boolean function
$f$ and a family of structured functions $\mathcal{Q}$. Throughout the manuscript $N$ will denote
the quantity $2^n$.
Constants $C$ may change from line to line without further notice.

We shall be using the following standard probabilistic bounds without further mention.


\begin{lemma}[Hoeffding bound for sampling \cite{TaoVu}]\label{lem:hoeffding-sample}
If $\bf X$ is a random variable with $\abs{{\bf X}} \leq 1$ and $\hat{\mu}$ is the empirical
average obtained from $t$ samples, then
\[ \prob{\abs{\ex{{\bf X}} - \hat{\mu}} ~>~ \gamma} ~\leq~ \exp(-\Omega(\gamma^2 t)) . \]
\end{lemma}

A Hoeffding-type bound can also be obtained for polynomial functions of $\pm 1$-valued random variables.


\begin{lemma}[Hoeffding bound for low-degree polynomials \cite{ODonnell08}]
\label{lem:hoeffding-low-degree}
Suppose that ${\bf F} = {\bf F}({\bf X}_1, \ldots, {\bf X}_N)$ is a polynomial of degree $d$ in random variables ${\bf
X}_1, \ldots, {\bf X}_N$ taking value $\pm 1$, then
\[ \prob{\abs{{\bf F} - \ex{{\bf F}}} ~>~ \gamma} ~\leq~ \exp\inparen{-\Omega\inparen{d \cdot
  \inparen{\gamma/\sigma}^{2/d}} } ,\]
where $\sigma = \sqrt{\ex{{\bf F}^2} - \ex{\bf F}^2}$ is the standard deviation of $\bf F$.
\end{lemma}

We start off by stating two fundamental results in additive combinatorics which are often applied in sequence. For a set $A\subseteq \F_2^n$, we write $A+A$ for the set of elements $a+a'$ such that $a, a' \in A$. More generally, the $k$-fold \emph{sumset}, denoted by $kA$, consists of all $k$-fold sums of elements of $A$.

First, the Balog-Szemer\'edi-Gowers theorem states that if a set has many additive quadruples, that is, elements $a_1,a_2,a_3,a_4$ such that $a_1+a_2=a_3+a_4$, then a large subset of it must have small sumset.


\begin{theorem}[Balog-Szemer\'edi-Gowers \cite{GSzT4}]\label{thm:bsg}
 Let $A\subseteq \F_2^{n}$ contain at least $|A|^3/K$ additive quadruples. Then there exists a subset $A' \subseteq A$ of size $|A'|\geq K^{-C}|A|$ with the property that $|A'+A'|\leq K^{C} |A'|$.
\end{theorem}

Freiman's theorem, first proved by Ruzsa in the context of $\F_2^n$, asserts that a set with small sumset is efficiently contained in a subspace.


\begin{theorem}[Freiman-Ruzsa Theorem \cite{ruzsa}]\label{thm:freiman}
Let $A \subseteq \F_2^n$ be such that $|A+A| \leq K |A|$. Then $A$ is contained in a subspace of size at most $2^{O(K^{C})}|A|$.
\end{theorem}

We shall also require the notion of a \emph{Freiman homomorphism}. We say the map $l$ is a Freiman 2-homomorphism if $x+y=z+w$ implies $l(x)+l(y)=l(z)+l(w)$. More generally, a Freiman homomorphism of order $k$ is a map $l$ such that $x_1+x_2+ \dots +x_k= x_1'+x_2'+ \dots +x_k'$ implies that $l(x_1)+ \dots + l(x_k)=l(x_1')+ \dots + l(x_k')$. The order of the Freiman homomorphism measures the degree of linearity of $l$; in particular, a truly linear map is a Freiman homomorphism of all orders.

Next we recall the definition of the uniformity of $U^k$ norms introduced by Gowers in \cite{GSzT4}.


\begin{definition}\label{uknorms}
Let G be any finite abelian group. For any positive integer $k \geq 2$
and any function $f: G \rightarrow \C$, define the \emph{$U^k$-norm}
by the formula
\[\|f\|_{U^k}^{2^k} = \E_{x,h_1, ..., h_k \in G} 
\prod_{\omega \in \{0,1\}^k} C^{|\omega|}f(x+\omega\cdot h),\]
where $\omega\cdot h$ is shorthand for $\sum_i\omega_ih_i$,
and $C^{|\omega|}f=f$ if $\sum_i\omega_i$ is even and $\overline{f}$
otherwise.
\end{definition}

In the special case $k=2$, a computation shows that
\[\|f\|_{U^2}=\|\widehat{f}\|_{l^4},\]
and hence any approach using the $U^{2}$ norm is essentially equivalent to using ordinary Fourier analysis. In the case $k=3$, the $U^3$ norm counts the number of additive octuples ``contained in" $f$, that is, we average over the product of $f$ at all eight vertices of a 3-dimensional parallelepiped in $G$.

These uniformity norms satisfy a number of important properties: they are clearly nested
\[\|f\|_{U^{2}} \leq \|f\|_{U^{3}} \leq \|f\|_{U^{4}} \leq ... \]
and can be defined inductively
\[\|f\|_{U^{k+1}}^{2^{k+1}}=\E_{x}\|f_{x}\|_{U^{k}}^{2^{k}},\]
where $k\geq 2$ and the function $f_{x}$ stands for the assignment $f_{x}(y)=f(y)\overline{f(x+y)}$. Thinking of the function $f$ as a complex exponential (a phase function), we can interpret the function $f_x$ as a kind of \emph{discrete derivative} of $f$.

It follows straight from a simple but admittedly ingenious sequence of applications of the
Cauchy-Schwarz inequality that if the balanced function $1_{A}-\alpha$ of a set $A \subseteq G$ of
density $\alpha$ has small $U^{k}$ norm, then $A$ contains the expected number of arithmetic
progressions of length $k+1$, namely $\alpha^{k+1}|G|^{2}$. This fact makes the uniformity norms
interesting for number-theoretic applications.

In computer science they have been used in the context of probabilistically checkable proofs (PCP)
\cite{SamorodnitskyT06}, communication complexity \cite{ViolaW07}, as well as in the analysis of
pseudo-random generators that fool low-degree polynomials \cite{BogdanovV10}.

In many applications, being small in the $U^{k}$ norm is a desirable property for a function to
have. What can we say if this is not the case? It is not too difficult to verify that
$\|f\|_{U^k}=1$ if and only if $f$ is a polynomial phase function of degree $k-1$, i.e. a function
of the form $\omega^{p(x)}$ where $p$ is a polynomial of degree $k-1$ and $\omega$ is an appropriate
root of unity. But does every function with large $U^{k}$ norm look like a polynomial phase function
of degree $k-1$?

It turns out that any function with large $U^{k}$ norm correlates, at the very least locally, with a polynomial phase function of degree $k-1$. This is known as the inverse theorem for the $U^k$ norm, proved by Green and Tao \cite{GrTu3} for $k=3$ and $p>2$ and Samorodnitsky \cite{Samorodnitsky07} for $k=3$ and $p=2$, and Bergelson, Tao and Ziegler \cite{BTZ,TZ} for $k>3$. We shall restrict our attention to the case $k=3$ in this paper, which we can state as follows.


\begin{theorem}[Global Inverse Theorem for $U^3$ \cite{GrTu3}, \cite{Samorodnitsky07}] \label{thm:globalinverse}
Let $f:\F_p^n\ra\C$ be a function such that
$\|f\|_\infty\leq 1$ and $\|f\|_{U^3}\geq\eps$. 
Then there exists a a quadratic form $q$ and a vector $b$ such that 
\[|\E_x f(x) \omega^{q(x)+b\cdot x}| \geq \exp(-O(\eps^{-C}))\]
\end{theorem}

In Section \ref{sec:refinement} we shall discuss various refinements of the inverse theorem, including correlations with so-called \emph{quadratic averages}. These refinements allow us to obtain polynomial instead of exponential correlation with some quadratically structured object.

We discuss further potential improvements and extensions of the arguments presented in this paper in Section \ref{sec:discussion}.

First of all, however, we shall turn to the problem of constructively obtaining a decomposition assuming that one has an efficient correlation testing procedure, which is done in Section \ref{sec:decompositions}.

\section{From decompositions to correlation testing}\label{sec:decompositions}
In this section we reduce from the problem of finding a decomposition for given function to the
problem of finding a single quadratic phase or average that correlates well with the function.

We state the basic decomposition result in somewhat greater generality as we believe it may be of
independent interest. We will consider a real-valued function $g$ on a finite domain $X$ (which shall be
$\F_2^n$ in the rest of the paper). We shall decompose the function $g$ in terms of members from an
arbitrary class $\cal Q$ of functions $\q: X \to [-1,1]$. $\cal Q$ may later be taken to be the class
of quadratic phases or quadratic averages. We will assume $\cal Q$ to be closed under negation of the
functions \ie $\q \in {\cal Q} \implies -\q \in {\cal Q}$. Finally, we shall consider a semi-norm 
$\snorm{\cdot}$ defined for functions on $X$, such that if $\snorm{f}$ is large for $f:X \to \R$
then $f$ has large correlation with some function in $\cal Q$. The obvious choice for
$\snorm{\cdot}$ is $\snorm{f} = \max_{\q \in {\cal Q}} \abs{\ip{f,\q}}$, as is the case in many known
decomposition results and the general result in \cite{TrevisanTV09}. However, we will be able to
obtain a stronger algorithmic guarantee by taking $\snorm{\cdot}$ to be the $U^3$ norm.


\begin{theorem}\label{thm:decomposition-general}
Let $\cal Q$ be a class of functions as above and let $\e, \delta > 0$ and $B > 1$.
Let $A$ be an algorithm which, given oracle access to a function 
$f: X \to [-B,B]$ satisfying $\snorm{f} \geq \e$, outputs, with probability at least $1-\delta$, a function
$\q \in {\cal Q}$ such that $\ip{f,\q} \geq \eta$ for some $\eta = \eta(\e,B)$. Then there exists an
algorithm which, given any
function $g: X \to [-1,1]$, outputs with probability at least $1-\delta/\eta^2$ a
decomposition  
\[ g~=~ c_1 \q_1 + \ldots + c_k \q_k + e + f \] 
satisfying $k \leq 1/\eta^2$,   $\snorm{f} \leq \epsilon$ and $\norm{e}_1 \leq 1/2B$. 
Also, the algorithm makes at most $k$ calls to $A$.
\end{theorem}

We prove the decomposition theorem building on an argument from \cite{TrevisanTV09}, which in
turn generalizes an argument of \cite{FriezeK99}. Both the arguments in \cite{TrevisanTV09,
FriezeK99} work well if for a function $f: X \to R$ satisfying 
$\max_{\q \in {\cal Q}}|\ip{f,\q}| \geq \e$, one can
efficiently find a $\q \in {\cal Q}$ with $\ip{f,\q} \geq \eta = \Omega(\e)$. 
It is important there that $\eta = \Omega(\e)$, or
 at least that the guarantee is independent of how $f$ is scaled.

Both proofs give an algorithm which, at each step $t$, checks if there exists $\q_t \in \cal Q$
which has good correlation with a given function $f_t$, and the decomposition is obtained by adding
the functions $\q_t$ obtained at
different steps. In both cases, the $\ell_{\infty}$ norm of the functions $f_t$ changes as the 
algorithm proceeds.

Suppose $\e' = o(\e)$ and we only had the scale-dependent guarantee that for functions 
$f:X \to [-1,1]$ with $\snorm{f} \geq \e$, we can efficiently find a $\q \in {\cal Q}$ such that $\ip{f,\q} \geq
\e^2$ (say). Then at step $t$ of the algorithm if we have $\norm{f_t}_{\infty} = M$ (say), then
$\snorm{f_t} \geq \e$ will imply $\snorm{f/M} \geq \e/M$ and one can only get a $\q_t$
satisfying $\ip{f_t,\q_t} \geq M \cdot (\e/M)^2 = \e^2/M$. Thus, the correlation of the functions $\q_t$ we
can obtain degrades as the $\norm{f_t}_{\infty}$ increases. This turns out to be insufficient to
bound the number of steps required by these algorithms and hence the number of terms in the
decomposition.

When testing correlations with quadratic phases using $\snorm{\cdot}$ as the $U^3$ norm, 
the correlation $\eta$ obtained for $f: \F_2^n \to [-1,1]$ has very bad
dependence on $\e$ and hence we run into the above problem. To get around it, we truncate the
functions $f_t$ used by the algorithm so that we have a uniform bound on their $\ell_{\infty}$
norms. However, this truncation introduces an extra term in the decomposition, for which we bound the
$\ell_1$ norm. Controlling the $\ell_1$ norm of this term requires a somewhat more sophisticated 
analysis than in \cite{FriezeK99}.  An analysis based on a similar potential function was also
employed in \cite{TrevisanTV09} (though not for the purpose of controlling the $\ell_1$ norm).

We note that a third term with bounded $\ell_1$
norm also appears in the (non-constructive) decompositions obtained in 
\cite{GW2}.

\begin{proofof}{of Theorem \ref{thm:decomposition-general}}
We will assume all calls to the algorithm $A$ correctly return a $q$ as above or declare
$\snorm{f} < \e$ as the case may be. The probability of any error in the calls to $A$ is at
most $k\delta$.

We build the decomposition by the following simple procedure.


\fbox{
\begin{minipage}{0.9\textwidth}
%
%
\begin{itemize}
\item[-] Define functions $f_1 = h_1 = g$. Set $t =1$.
\item[-] While $\snorm{f_t} \geq \epsilon$

\begin{itemize}
\item Let $\q_t$ be the output of $A$ when called with the function $f_t$. 
\item $h_{t+1} := h_t - \eta \q_t$.
\item $f_{t+1} := \trunc{h_{t+1}} = \max\{-B, \min\{ B, h_{t+1}\}\}$
\item $t := t+1$  
\end{itemize}
\end{itemize}

\end{minipage}
}
\smallskip

If the algorithm runs for $k$ steps, the decomposition it outputs is 
\[g = \sum_{t=1}^k \eta \cdot \q_t ~+~ (h_k - f_k)  ~+~ f_k \]
where we take $f = f_k$ and $e = h_k - f_k$. By construction, we have that $\snorm{f_k} \leq
\e$. It remains to show that $k \leq 1/\eta^2$ and $\norm{h_k - f_k}_1 \leq 1/2B$.

To analyze $\norm{h_t - f_t}$, we will define an additional function $\Delta_t \defeq f_t \cdot (h_t -
f_t)$. Note that $\Delta_t(x) \geq 0$ for every $x$, since $f_t$ is simply a truncation of $h_t$ and
hence $f_t = B$ when $h_t > f_t$ and $-B$ when $h_t < f_t$. This gives
\[  \norm{\Delta_t}_1 ~=~ \ex{\Delta_t} ~=~ \ex{f_t \cdot (f_t - h_t)} ~=~ \ex{B \cdot \abs{h_t - f_t}}
~=~ B \cdot \norm{h_t-f_t}_1 .\]

We will in fact bound the $\ell_1$ norm of $\Delta_k$ to obtain the required bound on
$\norm{h_k-f_k}_1$.  The following lemma states the bounds we need at every step.


\begin{lemma}\label{lem:step-bound}
For every input $x$ and every $t \leq k-1$
\[ f_t^2(x) - f_{t+1}^2(x) + 2\Delta_t(x) - 2\Delta_{t+1}(x) + \eta^2 
~\geq~ 2\eta \cdot \q_t(x) f_t(x) .\]
\end{lemma}

We first show how the above lemma suffices to prove the theorem. Taking expectations on both sides of the
inequality gives, for all $t \leq k-1$,
\[ \norm{f_t}_2^2 - \norm{f_{t+1}}_2^2 + 2\norm{\Delta_t}_1 - 2\norm{\Delta_{t+1}}_1 + \eta^2 
~\geq~ 2 \eta \cdot \ip{\q_t, f_t} ~\geq~ 2\eta^2 .\]
Summing over all $t \leq  k-1$ gives
\begin{equation*}
\norm{f_1}_2^2 - \norm{f_{k}}_2^2 + 2\norm{\Delta_1}_1 - 2\norm{\Delta_{k}}_1 ~\geq~ k \cdot
\eta^2 
~\Longrightarrow~
k \cdot \eta^2 + \norm{f_k}_2^2 + 2\norm{\Delta_k}_1 ~\leq~ 1
\end{equation*}
since $\norm{f_1}_2^2 = \norm{g}_2^2 \leq 1$ and $\Delta_1 = 0$. However, this gives 
$k \leq 1/\eta^2$ and $\norm{\Delta_k}_1 \leq 1/2$, 
which in turn implies $\norm{h_k-f_k}_1 \leq 1/2B$, completing the proof of Theorem \ref{thm:decomposition-general}.
\end{proofof}

We now return to the proof of Lemma \ref{lem:step-bound}.
\begin{proofof}{of Lemma \ref{lem:step-bound}}
We shall fix an input $x$ and consider all functions only at $x$. We
start by bringing the RHS into the desired form and collecting terms.
\begin{align*}
2\eta \q_t \cdot f_t &~=~  2(h_t - h_{t+1}) \cdot f_t \\
&~=~ 2(h_t - f_t) \cdot f_t - 2(h_{t+1} - f_{t+1}) \cdot f_{t+1} + 2f_t^2 - 2f_{t+1}^2 
- 2h_{t+1} \cdot f_t + 2 h_{t+1} \cdot f_{t+1} \\
&~=~ 2\Delta_t - 2\Delta_{t+1} + f_t^2 - f_{t+1}^2 + \inparen{f_t^2 - f_{t+1}^2 -2h_{t+1} (f_t - f_{t+1})}
\end{align*}
It remains to show that $f_t^2 - f_{t+1}^2 -2h_{t+1} (f_t - f_{t+1}) 
= (f_t - f_{t+1}) (f_t + f_{t+1} - 2h_{t+1})
\leq \eta^2$. We first note that if $\abs{f_{t+1}} < B$, then $h_{t+1} = f_{t+1}$ and the expression
becomes $(f_t - f_{t+1})^2$, which is at most $\eta^2$. Also, if $|f_t| = |f_{t+1}| = B$, then $f_t$ and $f_{t+1}$ must be equal (as $f_t$ only changes in steps of $\eta$) and the expression is 0.

Finally, in the case when $|f_t| < B$ and $|f_{t+1}| = B$, we must have that 
$\abs{f_t - h_{t+1}} = \abs{h_t - h_{t+1}} \leq \eta$. We can then bound the expression as
\[ (f_t - f_{t+1}) (f_t + f_{t+1} - 2h_{t+1}) ~\leq~ \inparen{\frac{(f_t - f_{t+1}) + (f_t + f_{t+1} -
    2h_{t+1})}{2}}^2
~=~ (f_t - h_{t+1})^2 ~\leq~ \eta^2,\]
which proves the lemma.
\end{proofof}

We next show that in the case when $\snorm{\cdot}$ is the $U^3$ norm and $\cal Q$ contains at most $\exp{(o(2^n))}$ functions, it is sufficient to test the correlations only for Boolean functions $f:
\F_2^n \to \{-1,1\}$. This can be done by simply scaling a function taking values in $[-B,B]$ to
$[-1,1]$ and then randomly rounding the value independently at each input to 
$\pm 1$ with appropriate probability.


\begin{lemma}\label{lem:reduceboolan}
Let $\eps,\d>0$. Let $A$ be an algorithm, which, given oracle access to a function $f: \F_2^n \to \{-1,1\}$ 
satisfying $\uthreenorm{f} \geq \e$, outputs, with probability at least
$1-\delta$, a function $\q \in {\cal Q}$ such that $\ip{f,\q} \geq \eta$ for some $\eta = \eta(\e)$.
In addition, assume that the running time of $A$ is $\poly(n, 1/\eta, \log(1/\delta))$.

Then there exists an algorithm $A'$ which, given oracle access to a function 
$f: \F_2^n \to [-B,B]$ satisfying $\uthreenorm{f} \geq \e$, outputs,
with probability at least $1-2\delta$, an element $\q \in {\cal Q}$ satisfying $\ip{f,\q} \geq \eta'$
for $\eta' = \eta'(\e,B)$. Moreover, the running time of $A'$ is $\poly(n, 1/\eta', \log(1/\delta))$.
\end{lemma}
\begin{proof}
Consider a random Boolean function $\tilde{f} : \F_2^n \to \{-1,1\}$ such that $\tilde{f}(x)$ is 1
with probability $(1 + f(x)/B)/2$ and $-1$ otherwise. $A'$ simply calls $A$ with the function
$\tilde{f}$ and parameters $\e/2B, \delta$. This means that whenever $A$ queries the value of
the function at $x$, $A'$ generates it independently of all other points by looking at
$f(x)$. It then outputs the $\q$ given by $A$.

If $\|\tilde{f}\|_{U^3} \geq \e/2B$, then $A$ outputs a $\q$ satisfying $\langle\tilde{f},\q\rangle \geq
\eta(\e/2B)$. If for the same $q$ we also have $\ip{f, \q} \geq B \cdot \eta(\e/2B) / 2 =
\eta'(\e,B)$, then the output of $A'$ is as desired. However, $\|\tilde{f}\|_{U^3}$ is a
polynomial of degree 8 and the correlation with any $\q$ is a linear polynomial in the $2^n$ random
variables $\{\tilde{f}(x)\}_{x \in \F_2^n}$. Thus, by Lemma \ref{lem:hoeffding-low-degree},
the probability that $\|\tilde{f}\|_{U^3} < \uthreenorm{f}/B - \e/2B$, or $\langle\tilde{f},\q\rangle \geq
\ip{f,\q}/B - \eta(\e/2B)/2$ for any $\q \in {\cal Q}$, is at most 
$\exp\inparen{-\Omega_{\e,B}\inparen{-|{\cal Q}| \cdot 2^n}} \leq \delta$.
\end{proof}

Thus, to compute the required decomposition into quadratic phases, one only needs to give an
algorithm for finding a phase $\q = (-1)^q$ satisfying $\ip{f,(-1)^q} \geq \eta$ when $f : \F_2^n \to \{-1,1\}$ 
is a \emph{Boolean} function satisfying $\uthreenorm{f} \geq \e$.

\section{Finding correlated quadratic phases over $\F_2^n$}
\label{sec:correlations}

In this section, we show how to obtain an algorithm for finding a quadratic phase which has good 
correlation with a given function Boolean $f:\F_2^n \to \{-1,1\}$ (if one exists). 
For an $f$ satisfying $\uthreenorm{f} \geq \e$, we want to find a quadratic form $q$ such
that $\ip{f,(-1)^q} \geq \eta(\e)$. 
The following theorem provides such a guarantee.


\begin{theorem}\label{thm:FindQuadratic}
Given $\e, \delta > 0$, there exists $\eta = \exp(-1/\e^{C})$ and
a randomized algorithm \FindQuadratic running in time $O(n^4 \log n \cdot \poly(1/\e, 1/\eta,
\log(1/\delta)))$ which, 
given oracle access to a function $f: \F_2^n \to \pmone$, either outputs a quadratic phase $(-1)^{q(x)}$
or $\bot$. The algorithm satisfies the following guarantee.
\begin{itemize}
\item If $\uthreenorm{f} \geq \e$, then with probability at
  least $1-\delta$ it finds a quadratic form $q$ such that $\ip{f,(-1)^q} \geq \eta$.
\item The probability that the algorithm outputs a quadratic form $q$ with $\ip{f,(-1)^q} \leq \eta/2$ is at 
most $\delta$.
\end{itemize}
\end{theorem}

The fact that $\uthreenorm{f} \geq \e$ implies the \emph{existence} of a quadratic phase $(-1)^q$ with
$\ip{f,(-1)^q} \geq \eta$ was proven by Samorodnitsky \cite{Samorodnitsky07}. We give an algorithmic 
version of his proof, starting with the proofs of the results from additive combinatorics contained therein.

Note that $\uthreenorm{f}^8$ is simply the expected value of the product 
$\prod_{\omega \in \B^3} f(x+\omega \cdot h)$ for random $x,h_1,h_2,h_3 \in \F_2^n$.
Hence,  Lemma \ref{lem:hoeffding-sample} implies that $\uthreenorm{f}$ can be easily estimated
by sampling sufficiently many values of $x,h_1,h_2,h_3$ and taking the average of the
products for the samples.

\begin{corollary}\label{cor:u3sample}
By making $O((1/\gamma^2) \cdot \log(1/\delta))$ queries to $f$, one can obtain an estimate
$\hat{U}$ such that
\[ \prob{|\uthreenorm{f}-\hat{U}| > \gamma} \leq \delta.\]
\end{corollary}

The main algorithm begins by checking if $\hat{U} \geq 3\e/4$ and rejects if this is not the case. If
$\hat{U} \geq 3\e/4$, then the above claim implies that $\uthreenorm{f} \geq \e/2$ with high
probability. So our algorithm will actually return a $q$ with correlation $\eta(\e')$ with $\e' = \e/2$. We shall
ignore this and just use $\e$ in the sequel for the sake of readability.

\subsection{Picking large Fourier coefficients in derivatives} \label{sec:large-Fourier}
The first step of the proof in \cite{Samorodnitsky07} is to find a choice function 
$\phi: \F_2^n \to \F_2^n$ which is ``somewhat linear''. The choice
function is used to pick a Fourier coefficient for the derivative $f_y$. 
The intuition is that if $f$ were indeed a quadratic phase of the form $(-1)^{\langle x, M x\rangle}$, then 
\[f_y(x) = f(x)f(x+y) = (-1)^{\langle x,(M+M^T) y\rangle} \cdot  (-1)^{\langle y, M y\rangle}  \ .\]
Thus, the largest Fourier coefficient (with absolute value 1) would be $\fyhat((M+M^T) y)$. 
Hence, there is a function 
$\phi(y) \defeq (M+M^T) y$, which is given by multiplying $y$ by a \emph{symmetric matrix}
$M + M^T$, which selects a large Fourier coefficient for $f_y$. 
The proof attempts to construct such a symmetric matrix for any $f$ with $\uthreenorm{f}
\geq \e$.

Expanding the $U^3$ norm and using \Holder's inequality gives the following lemma.


\begin{lemma}[Corollary 6.6 \cite{Samorodnitsky07}]
Suppose that $f: \F_2^n \ra \{-1,1\}$ is such that $\uthreenorm{f} \geq \e$. Then

\[\Ex{x,y}{\sum_{\alpha,\beta} \fxhat^2(\alpha) \cdot \fyhat^2(\beta) \cdot
  \fxyhat^2(\alpha+\beta)} ~\geq~ \e^{16}.\]
\end{lemma}
Choosing a random function $\phi(x) = \alpha$ with 
probability $\fxhat^2(\alpha)$ satisfies
\[ \Prob{x,y}{\phi(x) + \phi(y) = \phi(x+y)} ~=~ 
\sum_{\alpha,\beta} \fxhat^2(\alpha) \cdot \fyhat^2(\beta) \cdot
  \fxyhat^2(\alpha+\beta) . \] 
Thus, when $\uthreenorm{f} \geq \e$ , the above lemma gives that
\[ \Prob{\phi,x,y}{\phi(x) + \phi(y) = \phi(x+y)} ~=~
\Ex{x,y}{\sum_{\alpha,\beta} \fxhat^2(\alpha) \cdot \fyhat^2(\beta) \cdot
  \fxyhat^2(\alpha+\beta)}
~\geq~ \e^{16} . \] 

The proof in \cite{Samorodnitsky07} works with a random function $\phi$
as described above. 
We define a slightly different random function $\phi$, since we need its value at any
input $x$ to be samplable in time polynomial in $n$. Thus, we will only sample $\alpha$ for which
the corresponding Fourier coefficients are sufficiently large. 
In particular, we need an algorithmic
version of the decomposition of a function into linear phases, which follows from the
Goldreich-Levin theorem.


\begin{theorem}[Goldreich-Levin \cite{GoldreichL89}]\label{thm:linear-decomposition}
Let $\gamma,\delta>0$. There is a randomized algorithm \LinearDecomposition, which, given oracle access to a function $f: \F_2^n \to \pmone$, runs in time
$O(n^2\log n \cdot \poly(1/\gamma, \log(1/\delta)))$ and outputs a decomposition 
\[f = \sum_{i=1}^k
c_i \cdot (-1)^{\langle \alpha_i,x\rangle} + f'\]
with the following guarantee:
\begin{itemize}
\item $k = O(1/\gamma^2)$.
\item $\prob{\exists i~ \smallabs{c_i - \fhat(\alpha_i)} > \gamma/2} \leq \delta$.
\item $\prob{\forall \alpha ~\text{such that}~
    \smallabs{\fhat(\alpha)} \geq \gamma,~~ \exists i~ \alpha_i = \alpha} \geq 1-\delta$.
\end{itemize}
\end{theorem}

\begin{remark}\label{rem:Goldreich-Levin}
Note that the above is a slightly non-standard version of the Goldreich-Levin theorem. The usual
one makes $O(n\log n \cdot \poly(1/\gamma, \log(1/\delta)))$ queries to
$f$ (where each query takes $O(n)$ time to write down) and guarantees that for any specific
$\alpha$ such that $\smallabs{\fhat(\alpha)} \geq \gamma$, there exists an $i$ with 
$\alpha_i = \alpha$, with probability at least $1-\delta$. By repeating the algorithm
$O(\log(1/\gamma))$ 
times, we can take a union bound over all $\alpha$ as in the last property guaranteed by the above theorem.
\end{remark}

It follows that in order to sample $\phi(x)$, instead of sampling from all Fourier coefficients of $f_x$, we only
sample from the large Fourier coefficients using the above decomposition. We shall denote
the quantity $\e^{16}/4$ that appears below by $\rho$.
%

\begin{lemma}\label{lem:sample-phi}
There exists a distribution over functions $\phi: \F_2^n \to \F_2^n$ such that $\phi(x)$ is
independently chosen for each $x \in \F_2^n$, and is samplable in time 
$O(n^3\log n \cdot \poly(1/\e))$  given oracle access to $f$. Moreover, if 
$\uthreenorm{f} \geq \e$, then we have
\[ \Prob{\phi}{ \Prob{x, y}{\phi(x) + \phi(y) = \phi(x+y)} \geq \e^{16}/4 } \geq \e^{16}/4 . \]
\end{lemma}
\begin{proof}
We sample $\phi(x)$ at each input $x$ as follows. We run \LinearDecomposition for $f_x$ 
with $\gamma = \delta = \e^{16}/18$ and sample $\phi(x)$ to
be $\alpha_i$ with probability $c_i^2$. If $\sum c_i^2 < 1$, we answer
arbitrarily with the remaining probability. By Theorem \ref{thm:linear-decomposition}, 
with probability at least 
$1- 2\gamma$ over the run of \LinearDecomposition, each $\alpha \in \F_2^n$ with 
$\smallabs{\fxhat(\alpha)} \geq \gamma$ is sampled with probability at least 
$(\fxhat(\alpha) - \gamma/2)^2 \geq \fxhat^2(\alpha) - \gamma$. Let $[z]_0$ denote
$\max\{0,z\}$. We have
\begin{align*}
\Prob{\phi, x, y}{\phi(x) + \phi(y) = \phi(x+y)}  &\geq 
\Ex{x,y}{\sum_{\alpha,\beta} (1-2\gamma)^3 \insquare{\fxhat^2(\alpha) -
    \gamma}_0 \insquare{\fyhat^2(\beta) - \gamma}_0 \insquare{\fxyhat^2(\alpha+\beta) - \gamma}_0} \\
&\geq 
\e^{16} - 9\gamma,
\end{align*}
which by our choice of parameters is at least $\eps^{16}/2$.
This immediately implies that
$\Prob{\phi}{ \Prob{x, y}{\phi(x) + \phi(y) = \phi(x+y)} \geq \e^{16}/4 } \geq \e^{16}/4$.
\end{proof}

Thus, with probability $\rho=\e^{16}/4$ one gets a good $\phi$ which is somewhat linear. This $\phi$ is then
used to recover an appropriate quadratic phase. We will actually delay sampling the
function on all points and only query $\phi(x)$ when needed in the construction of the quadratic
phase (which we show can be done by querying $\phi$ on polynomially many points). Consequently, the
construction procedures that follow will only work with a small probability, i.e. when we are actually
working with a good $\phi$. However, we can test the quadratic phase we obtain in the end and repeat
the entire process if the phase does not correlate well with $f$. Also, note that we store the
$(x,\phi(x))$ already sampled in a data structure and re-use them if and when the same $x$ is queried
again.

\subsection{Applying the Balog-\Szemeredi-Gowers theorem} \label{sec:bsg}
The next step of the proof uses $\phi$ to obtain a linear choice function $Dx$ 
for some matrix $D$. This step uses certain results from additive combinatorics, 
for which we develop algorithmic versions below. In particular, it applies the 
Balog-\Szemeredi-Gowers (BSG) theorem to the set 
\[A_{\phi} \defeq \inbraces{(x,\phi(x)) \suchthat \smallabs{\fxhat(\phi(x))} \geq \gamma},\]
where we will choose $\gamma = O(\e^{16})$ as in Lemma \ref{lem:sample-phi}.

For any set $A \in \B^{n}$ that is somewhat linear, the 
Balog-\Szemeredi-Gowers theorem allows us to find a subset $A' \subseteq A$ which is large 
and does not grow too much when added to itself. We state the following version from
\cite{BalogS94}, which is particularly suited to our application.

%
\begin{theorem}[Balog-Szemer\'edi-Gowers Theorem \cite{BalogS94}]\label{thm:bsg2}
Let $A \subseteq \F_2^n$ be such that $\Prob{a_1, a_2 \in A}{a_1 + a_2 \in A} \geq \rho$.
Then there exists $A' \subseteq A$, $|A|' \geq \rho |A|$ such that 
$|A'+ A'| \leq (2/\rho)^{8} |A|$.
\end{theorem}
We are interested in finding the set $A_{\phi}'$ which results from applying the 
above theorem to the set $A_{\phi}$. 
However, since the set $A_{\phi}'$ is of exponential size, we do not have time to
write down the entire set (even if we can find it). Instead, we will need an efficient 
algorithm for testing membership in the set.
To get the required algorithmic version, we follow the proof by 
Sudakov, \Szemeredi and Vu \cite{SudakovSV05} and the presentation by Viola \cite{Viola07}.

In this proof one actually constructs a graph on the set $A_{\phi}$ and then selects a subset of 
the neighborhood of a random vertex as $A_{\phi}'$, after removing certain problematic vertices.
It can be deduced that the set $A_{\phi}'$ can be found in time polynomial in the 
size of the graph. However, as discussed above, this is 
still \emph{exponential} in $n$ and hence inadequate for our purposes.
Below, we develop a test to check if a certain element $(x,\phi(x))$ is in $A_{\phi}'$.

We first define a (random) graph on the vertex set
\footnote{Since $\phi$ is random, the vertex set
  of the graph as defined is random. However, since $\phi$ is a function, the vertex set is
  isomorphic to $\F_2^n$ and one may think of the graph as being defined on a fixed set of vertices with
  edges chosen according to a random process.} 
$\inbraces{(x,\phi(x)) \mid x \in \F_2^n}$ and edge set $E_{\gamma}$ for $\gamma > 0$, defined as
\[
E_{\gamma} ~\defeq~ \inbraces{(x,\phi(x)), (y,\phi(y)) ~\left\lvert~
\begin{array}{c}
\phi(x) + \phi(y) = \phi(x+y) \\  ~\text{and}~ \\
\smallabs{\fxhat(\phi(x))}, \smallabs{\fyhat(\phi(y))}, 
\smallabs{\fxyhat(\phi(x+y))} \geq \gamma
\end{array}
\right.}.\]

Lemma \ref{lem:sample-phi} implies that over the choice of $\phi$, with probability 
at least $\rho=\e^{16}/4$, the graph defined with $\gamma = \e^{16}/18$, has density at 
least $\rho$. However, if a $\phi$ is good for a certain 
value of $\gamma$, then it is also good for all values $\gamma' \leq \gamma$ (as
the density of the graph can only increase). For the remaining argument, we will assume 
that we have sampled $\phi$ completely and that it is good. We will later choose
$\gamma \in [\e^{16}/180, \e^{16}/18]$.

Since we will be examining the properties of certain neighborhoods in this graph, we first write a procedure 
to test if two vertices in the graph have an edge between them.

\fbox{
\begin{minipage}{0.9\textwidth}

\smallskip

\edgetest(u,v,$\gamma$)
\begin{itemize}
\item[-] Let $u = (x,\phi(x))$ and $v=(y,\phi(y))$. 
\item[-] Estimate $\smallabs{\fxhat(\phi(x))}, \smallabs{\fyhat(\phi(y))}$ and 
$\smallabs{\fxyhat(\phi(x+y))}$ using $t$ samples for each.
\item[-] Answer 1 if $\phi(x)+\phi(y) = \phi(x+y)$ and 
all estimates are at least $\gamma$, and 0 otherwise. 
\end{itemize}

\end{minipage}
}
\smallskip

Unfortunately, since we are only estimating the Fourier coefficients, we will only be able to 
test if two vertices have an edge between them with a slight error
in the threshold $\gamma$, and with high probability. Thus, if the estimate is at least $\gamma$,
we can only say that with high probability, the Fourier coefficient must be at least 
$\gamma - \gamma'$ for a small error $\gamma'$. This leads to the following guarantee on
\edgetest.
%

\begin{claim}\label{clm:edgetest}
Given $\gamma', \delta > 0$, the output of \edgetest($u,v,\gamma$) with 
$t = O(1/\gamma'^2 \cdot \log(1/\delta))$ queries, satisfies the following
guarantee with probability at least $1-\delta$.
\begin{itemize}
\item $\edgetest(u,v,\gamma) = 1  ~\Longrightarrow (u,v)~ \in E_{\gamma - \gamma'}$.
\item $\edgetest(u,v,\gamma) = 0  ~\Longrightarrow (u,v)~ \notin E_{\gamma + \gamma'}$.
\end{itemize}
\end{claim}
\begin{proof}
The claim follows immediately from Lemma \ref{lem:hoeffding-sample} and the definitions of
$E_{\gamma - \gamma'}$, $E_{\gamma + \gamma'}$.
\end{proof}

The approximate nature of the above test introduces a subtle issue. Note that the outputs
1 and 0 of the test correspond to the presence or absence of edges in \emph{different graphs} with
edge sets $E_{\gamma-\gamma'}$ and $E_{\gamma+\gamma'}$. 
The edge sets of the two graphs are related as
$E_{\gamma+\gamma'} \subseteq E_{\gamma-\gamma'}$. 
But the proof of Theorem \ref{thm:bsg2} uses somewhat more complicated
subsets of vertices, which are defined using both upper and lower bounds on the
sizes of certain neighborhoods. Since the upper and lower bounds estimated using
the above test will hold for slightly different graphs, we need to be
careful in analyzing any algorithm that uses \edgetest as a primitive.

We now return to the argument as presented in \cite{SudakovSV05}. It considers the neighborhood 
of a random vertex $u$ and removes
vertices that have too few neighbors in common with other vertices in the graph. 
Let the size of the vertex set be $N = 2^n$. 
For a vertex $u$, we define the following sets:
\begin{align*}
N(u) &\defeq~ \inbraces{v \suchthat (u,v) \in E_{\gamma}}\\
S(u) &\defeq~ \inbraces{v \in N(u) \suchthat \Prob{v_1}{v_1 \in N(u) ~\text{and}~\abs{N(v) \cap N(v_1)}
    \leq \rho^3 N} \geq \rho^2}\\
     &=~~ \inbraces{v \in N(u) \suchthat \Prob{v_1}{v_1 \in N(u) ~\text{and}~ 
     \Prob{v_2}{v_2 \in N(v) \cap N(v_1)} \leq \rho^3} > \rho^2}\\
T(u) &\defeq~ N(u) \setminus S(u) \\
     &=~~ \inbraces{v \in N(u) \suchthat \Prob{v_1}{v_1 \in N(u) ~\text{and}~ 
     \Prob{v_2}{v_2 \in N(v) \cap N(v_1)} \leq \rho^3} \leq \rho^2}
\end{align*}

It is shown in \cite{SudakovSV05} (see also \cite{Viola07}) that if the graph has density $\rho$,
then picking $A_{\phi}' = T(u)$ for a
random vertex $u$ is a good choice\footnote{Note that here we are choosing $A_{\phi}'$ to be the
neighborhood of \emph{any} vertex in the graph, instead of vertices in $A_{\phi}$. However, 
this is not a problem since the only vertices with non-empty neighborhoods are the ones in
$A_{\phi}$.}.


\begin{lemma}\label{lem:bsg-ssv}
Let the graph with edge set $E_{\gamma}$ have density at least $\rho$ and let $A_{\phi}' = T(u)$
for a random vertex $u$. Then, with probability at least $\rho/2$ over the choice of $u$, 
the set $A_\phi'$ satisfies
\[ \abs{A_{\phi}'} \geq \rho N \quad \text{and} 
\quad \abs{A_{\phi}' + A_{\phi}'} \leq (2/\rho)^8 N.\]
\end{lemma}

We now translate the condition for membership in the set $T(u)$ into an algorithm. Note that
we perform different edge tests with different thresholds, the values of which will be chosen later.


\fbox{
\begin{minipage}{0.9\textwidth}

\smallskip

\bsgtest ($u, v, \gamma_1, \gamma_2, \gamma_3, \rho_1, \rho_2$) \hspace{1.1 in} (Approximate test to check if $v \in T(u)$)
\begin{itemize}
\item[-] Let $u=(x,\phi(x))$ and $v=(y,\phi(y))$.
\item[-] Sample $(z_1, \phi(z_1)), \ldots, (z_r, \phi(z_r))$.
\item[-] For each $i \in [r]$, sample $(w_1^{(i)}, \phi(w_1^{(i)})), \ldots, (w_s^{(i)}, \phi(w_s^{(i)}))$.
\item[-] If \edgetest(u,v,$\gamma_1$) = 0, then output 0.
\item[-] For $i \in [r], j \in [s]$,  let 
\begin{align*}
X_i   &~=~ \edgetest\inparen{(x,\phi(x)), (z_i,\phi(z_i)),\gamma_2}\\
Y_{ij} &~=~ \edgetest\inparen{(y,\phi(y)), \inparen{w_j^{(i)}, \phi\inparen{w_j^{(i)}}},\gamma_3}\\
Z_{ij} &~=~ \edgetest\inparen{(z_i,\phi(z_i)), \inparen{w_j^{(i)}, \phi\inparen{w_j^{(i)}}},\gamma_3}
\end{align*}
\item[-] For each $i$, take $B_i = 1$ if $\frac{1}{s} \sum_j Y_{ij} \cdot Z_{ij} \leq \rho_1$ and 0
  otherwise.
\item[-] Answer 1 if $\frac{1}{r} \sum_i X_i \cdot B_i \leq \rho_2$ and 0 otherwise.
\end{itemize}

\end{minipage}
}
\smallskip

\paragraph{Choice of parameters for \bsgtest:}
We shall choose the parameters for the above test as follows. Recall that $\rho = \e^{16}/4$.
We take $\rho_1 = 21\rho^3/20$ and $\rho_2 = 19\rho^2/20$. Given an error parameter
$\delta$, we take $r$ and $s$ to be $\poly(1/\rho, \log(1/\delta))$, so that with 
probability at least $1-\delta$, the error in the last two estimates is at most $\rho^3/100$.
Also, by using $\poly(1/\rho, \log(1/\delta))$ samples in each call to \edgetest, we can assume
that the error in all estimates used by \edgetest is at most $\rho^3/100$.

To choose $\gamma_1,\gamma_2,\gamma_3$, we divide the interval $[\e^{16}/180,\e^{16}/18]$ into
$4/\rho^2$ consecutive sub-intervals of size $\rho^3/20$ each. We then randomly choose a
sub-interval and choose positive 
parameters $\gamma, \mu$ so that $\gamma-\mu$ and $\gamma+\mu$ are endpoints of
this interval. We set $\gamma_1 = \gamma_3 = \gamma + \mu/2$ and $\gamma_2 = \gamma - \mu/2$.

\medskip
To analyze \bsgtest, we ``sandwich''  the elements on which it answers 1 between a
large set and a set with small doubling.
%
%
\begin{lemma}\label{lem:bsgtest}
Let $\delta > 0$ and parameters $\rho_1, \rho_2, r, s$ 
be chosen as above. Then for every $u = (x,\phi(x))$ and every choice of $\gamma_1, \gamma_2, \gamma_3$ as above, there exist two sets $A_{\phi}^{(1)}(u) \subseteq A_{\phi}^{(2)}(u)$, 
such that the output of \bsgtest satisfies the following
with probability at least $1-\delta$.
\begin{itemize}
\item $\bsgtest(u,v,\gamma_1, \gamma_2, \gamma_3, \rho_1, \rho_2) = 1 
\quad\Longrightarrow\quad v \in A_{\phi}^{(2)}(u)$.
\item $\bsgtest(u,v,\gamma_1, \gamma_2, \gamma_3, \rho_1, \rho_2) = 0 
\quad\Longrightarrow\quad v \notin A_{\phi}^{(1)}(u)$.
\end{itemize}
Moreover, with probability $\rho^3/24$ over the choice of $u$ and $\gamma_1,\gamma_2,\gamma_3$, 
we have
\[ |A_{\phi}^{(1)}(u)| \geq (\rho/6) \cdot N \quad \text{and} 
\quad |A_{\phi}^{(2)}(u) + A_{\phi}^{(2)}(u)| \leq (2/\rho)^8 \cdot N.\]
\end{lemma}
\begin{proof}
To deal with the approximate nature of \edgetest, we define the following sets:
\begin{align*}
N_{\gamma}(u) &\defeq \inbraces{v \suchthat (u,v) \in E_{\gamma}}\\
T(u,\gamma_1,\gamma_2,\gamma_3,\rho_1,\rho_2) &\defeq
\inbraces{v \in N_{\gamma_1}(u) \suchthat \Prob{v_1}{v_1 \in N_{\gamma_2}(u) ~\&~ 
     \Prob{v_2}{v_2 \in N_{\gamma_3}(v) \cap N_{\gamma_3}(v_1)} \leq \rho_1} \leq \rho_2}
\end{align*}
Going through the definitions and recalling that $E_{\gamma} \subseteq E_{\gamma - \gamma'}$ 
for $\gamma' > 0$, it can be checked that the sets 
$T(u, \gamma_1,\gamma_2,\gamma_3,\rho_1,\rho_2)$ are monotone in the various parameters. 
In particular, for $\gamma_1', \gamma_2', \gamma_3', \rho_1', \rho_2' > 0$
\[ T(u, \gamma_1,\gamma_2,\gamma_3,\rho_1,\rho_2) ~\subseteq~ 
T(u, \gamma_1 - \gamma_1',\gamma_2 + \gamma_2',\gamma_3-\gamma_3',\rho_1-\rho_1',\rho_2+\rho_2').\]

Recall that we have $\gamma_1 = \gamma_3 = \gamma + \mu/2$ and $\gamma_2 = \gamma - \mu/2$,
where $[\gamma-\mu,\gamma+\mu]$ is a sub-interval of $[\e^{16}/180,\e^{16}/18]$ of length $\rho^3/20$.


We define the sets $A_{\phi}^{(1)}(u)$ and $A_{\phi}^{(2)}(u)$ as below.
\begin{align*}
A_{\phi}^{(1)}(u) &~\defeq~ T(u,\gamma+\mu, \gamma-\mu, \gamma+\mu, 11\rho^3/10, 9\rho^2/10) \\
A_{\phi}^{(2)}(u) &~\defeq~ T(u,\gamma, \gamma, \gamma, \rho^3, \rho^2)
\end{align*}
By the monotonicity property noted above, we have that $A_{\phi}^{(1)}(u) \subseteq
A_{\phi}^{(2)}(u)$. Also, by the choice of parameters $r$, $s$ and the number of samples in 
\edgetest, we know that with probability $1-\delta$, the error in all estimates used in
\bsgtest is at most $\rho^3/100$. Hence, we get that with probability at least $1-\delta$,
if \bsgtest answers 1, then the input is in $A_{\phi}^{(2)}$ and if \bsgtest answers 0, then 
it is not in $A_{\phi}^{(1)}$. It remains to prove the bounds on the size and doubling of
these sets.

By our choice of parameters, $A_{\phi}^{(2)}(u)$ is the same set as the one defined in Sudakov \etal 
\cite{SudakovSV05}.  They show that if $u$ is such that $\smallabs{A_{\phi}^{(2)}(u)} \geq
3 \cdot (\rho/2)^2 N$, then 
$\smallabs{A_{\phi}^{(2)}(u) + A_{\phi}^{(2)}(u)} \leq (2/\rho)^8 \cdot N$ (see Lemma 3.2 in
\cite{Viola07}  for a simplified proof of the version mentioned here). To show the lower bound on
the size of $A_{\phi}^{(2)}(u)$, we will show that in fact with probability at least $\rho^3/24$
over the choice of $u$ and $\gamma_1,\gamma_2,\gamma_3$, we will have 
$\smallabs{A_{\phi}^{(1)}(u)} \geq (\rho/6) \cdot N$. 
Since $A_{\phi}^{(1)}(u) \subseteq A_{\phi}^{(2)}(u)$, this suffices for the proof.

We consider a slight modification of the argument of \cite{SudakovSV05}, showing 
an upper bound on the expected size of the set $S'(u)$ defined as
\begin{align*} 
S'(u) &\defeq
N_{\gamma + \mu}(u) \setminus T(u, \gamma+\mu, \gamma - \mu, \gamma+\mu, 11\rho^3/10, 9\rho^2/10)\\
&= \inbraces{v \in N_{\gamma + \mu}(u) \suchthat \Prob{v_1}{v_1 \in N_{\gamma - \mu}(u) ~\&~ 
     \Prob{v_2}{v_2 \in N_{\gamma + \mu}(v) \cap N_{\gamma + \mu}(v_1)} \leq 11\rho^3/10} \geq
   9\rho^2/10} .
\end{align*}
We know from Lemma \ref{lem:sample-phi}  that since $\gamma+\mu \leq \e^{16}/18$, the quantity
$\Ex{u}{|N_{\gamma+\mu}(u)|}$, which is the average degree of the graph, is at least
$\rho N$ (assuming that we are working with a good function $\phi$). Combining this
with an upper bound on $\Ex{u}{|S'(u)|}$ will give the required lower bound on the size of
$A_{\phi}^{(1)}(u)  =  
T(u, \gamma+\mu, \gamma - \mu, \gamma+\mu, 11\rho^3/10, 9\rho^2/10)$.

We call a pair $(v,v_1)$ \emph{bad} if 
$\abs{N_{\gamma+\mu}(v) \cap N_{\gamma+\mu}(v)} \leq 11\rho^3N/10$.
We need the following bound.

\begin{claim}\label{clm:badpairs}
There exists a choice for the sub-interval 
$[\gamma-\mu,\gamma+\mu]$  of length $\rho^3/20$ in $[\e^{16}/180, \e^{16}/18]$
such that
\[\Ex{u}{\#\inbraces{\text{bad pairs}~ (v,v_1) 
~:~ v \in N_{\gamma + \mu}(u) ~\&~ v_1 \in N_{\gamma - \mu}(u)}} ~\leq~ 3\rho^3N^2/5\]
\end{claim}

We first prove Lemma \ref{lem:bsgtest} assuming the claim. From the definition of $S'(u)$,
\[\#\{\text{bad pairs}~ (v,v_1) 
~:~ v \in N_{\gamma + \mu}(u) ~\&~ v_1 \in N_{\gamma - \mu}(u)\} ~\geq~ 
|S'(u)| \cdot (9\rho^2N/10).\]
Claim \ref{clm:badpairs} gives $\Ex{u}{|S'(u)|} \leq (3\rho^3N^2/5)/(9\rho^2N/10)  =  
(2\rho/3)N$, for at least one choice of the interval $[\gamma-\mu,\gamma+\mu]$.
Since there are $4/\rho^2$
choices for the sub-interval, this happens with probability at least $\rho^2/4$. 

For this choice of $\gamma$ and $\mu$ (and hence of $\gamma_1,\gamma_2,\gamma_3$),
we also have $\Ex{u}{\smallabs{N_{\gamma+\mu}(u)}} \geq \rho N$. Since 
$S'(u) =N_{\gamma + \mu}(u) \setminus A_\phi^{(1)}$, we get that 
$\Ex{u}{\smallabs{A_\phi^{(1)}}} \geq \rho N - (2\rho/3) N = (\rho/3) N$.
Hence, with probability at least $\rho/6$ over the choice of $u$, 
$\smallabs{A_\phi^{(1)}} \geq (\rho/6) N$. Thus, we obtain the desired outcome with
probability at least $\rho^3/24$ over the choice of $u$ and $\gamma_1,\gamma_2,\gamma_3$.
\end{proof}
\begin{proofof}{of Claim \ref{clm:badpairs}}
We begin by observing that the expected number of bad pairs $(v,v_1)$ such that $v \in N_{\gamma +
  \mu}(u) ~\&~ v_1 \in N_{\gamma - \mu}(u)$ is equal to
\begin{eqnarray*}
&\Ex{u}{\#\inbraces{\text{bad pairs}~ (v,v_1) 
~:~ v \in N_{\gamma + \mu}(u) ~\&~ v_1 \in N_{\gamma + \mu}(u)}}\\ &+ 
\Ex{u}{\#\inbraces{\text{bad pairs}~ (v,v_1) 
~:~ v \in N_{\gamma + \mu}(u) ~\&~ v_1 \in N_{\gamma - \mu} (u)\setminus N_{\gamma + \mu}(u)}}.
\end{eqnarray*}
Note that for each of the $\binom{N}{2}$
choices for ${v,v_1}$, if they form a bad pair, then each $u$ is in 
$N_{\gamma + \mu}(v) \cap N_{\gamma + \mu}(v_1)$ with probability
at most $11\rho^3/10$. Hence, the first term is at most $(11\rho^3/20) N^2$.
Also, the second term is at most
\[ N \cdot \Ex{u}{\abs{N_{\gamma-\mu}(u) \setminus N_{\gamma + \mu}(u)}}
~=~ N \cdot \inparen{\Ex{u}{\abs{N_{\gamma-\mu}(u)}} - \Ex{u}{\abs{N_{\gamma+\mu}(u)}} } \]
We know that $\Ex{u}{\abs{N_{\gamma}(u)}}$ is monotonically decreasing in $\gamma$.
Since it is at most $N$ for $\gamma = \e^{16}/180$, there is at least one interval of size 
$\rho^3/20$ in $[\e^{16}/180, \e^{16}/18]$, where the change is at most $\rho^3 N/20$.
Taking $\gamma + \mu$ and $\gamma - \mu$ to be the endpoints of this interval finishes 
the proof.
\end{proofof}

\subsection{Obtaining a linear choice function} \label{sec:lin-choice}
Using the subset given by the Balog-\Szemeredi-Gowers theorem, one can use the somewhat linear
choice function
$\phi$ to find an \emph{linear transformation}  
$x \mapsto Tx$ which also selects large Fourier coefficients in derivatives. 
In particular, it satisfies $\Ex{x}{\fxhat^2(Tx)} \geq \eta$ for some $\eta = \eta(\e)$.
This map $T$ can then be used to find an appropriate quadratic phase. 

In this subsection, we  give an algorithm for finding such a transformation, using the procedure
\bsgtest developed above. In the lemma below, we assume as before that $\phi$ is a good function satisfying
the guarantee in Lemma \ref{lem:sample-phi}. We also assume that we have chosen a good vertex $u$ 
and parameters $\gamma_1,\gamma_2,\gamma_3$ satisfying the
guarantee in Lemma \ref{lem:bsgtest}.


\begin{lemma}\label{lem:lin-choice}
Let $\phi$ be as above and $\delta > 0$. Then there exists an $\eta = \exp(-1/\e^{C})$ and an 
algorithm which makes $O(n^2 \log n \cdot \poly(1/\eta, \log(1/\delta)))$ calls to \bsgtest 
and uses additional running time $O(n^3)$ to
output a linear map $T$ or the symbol $\bot$. 
If \bsgtest is defined using a good $u$ and parameters $\gamma_1,\gamma_2,\gamma_3$ 
as above, then with probability at
least $1-\delta$ the algorithm outputs a map $T$ satisfying $\Ex{x}{\fxhat^2(Tx)} \geq \eta$.
\end{lemma}
\begin{proof}
Let $t = 4n^2 + \log(10/\delta)$.
We proceed by first sampling $K = 100t/\rho$ 
elements $(x,\phi(x))$ and running 
\bsgtest($u, \cdot$) on each of them with parameters as in Lemma \ref{lem:bsgtest} and 
$\delta' = \delta/(5K)$. We retain only the points $(x,\phi(x))$ on which \bsgtest outputs 1.
Since $\delta' = \delta/(5K)$, $\bsgtest$ does not satisfy the guarantee of Lemma 
\ref{lem:bsgtest} on some query with probability at most $\delta/5$. We assume this does not
happen for any of the points we sampled.

If \bsgtest outputs 1 on fewer than $t$ of the queries, we stop and output $\bot$. The following
claim shows that the probability of this happening is at most $\delta/5$. In fact, the claim shows
that with probability $1-\delta/5$ there must be at least $t$ samples from 
$A_{\phi}^{(1)}$ itself, on which we assumed that \bsgtest outputs 1.


\begin{claim}\label{clm:unisamples}
With probability at least $1-\delta/5$, the sampled points contain at least $t$ samples 
from $A_{\phi}^{(1)}$.
\end{claim}
\begin{proof}
Since $|A_{\phi}^{(1)}| \geq \rho N/6$, the expected number of samples from $A_{\phi}^{(1)}$
is at least $\rho K/6$. By a Hoeffding bound, the probability that this number is less than $t$ is at most
$\exp(-\Omega(\rho K)) \leq \delta/5$ if $\rho K = \Omega(\log(1/\delta))$.
\end{proof}

Note that conditioned on being in $A_{\phi}^{(1)}$, the sampled points are in fact \emph{uniformly}
distributed in $A_{\phi}^{(1)}$. We show that then they must span a subspace of large dimension, and
that their span must cover at least half of $A_{\phi}^{(1)}$.


\begin{claim}\label{clm:large-span}
Let $z_1, \ldots, z_t \in A_{\phi}^{(1)}$ be uniformly sampled points. Then for 
$t \geq 4n^2 + O(\log(1/\delta))$ it is true with
probability $1-\delta/5$ that
\begin{itemize}
\item $|<z_1, \ldots, z_t> \cap A_{\phi}^{(1)}| \geq (1/2) |A_{\phi}^{(1)}|$
\item $\dim(<z_1, \ldots, z_t>) \geq n - \log(12/\rho)$.
\end{itemize}
\end{claim}
\begin{proof}
For the first part, we consider the span $<z_1, \ldots, z_t>$, which is a subspace of $\F_2^n$. The 
probability that it has small intersection with $A_{\phi}^{(1)}$ is 
\[\sum_{|S \cap A_{\phi}^{(1)}| \leq |A_{\phi}^{(1)}|/2} \prob{z_1, \ldots, z_t \in S} \cdot
\prob{<z_1, \ldots, z_t> ~=~ S \mid z_1, \ldots, z_t \in S},\]
where the sum is taken over all
subspaces $S$ of $\F_2^n$.
Since $|S \cap A_{\phi}^{(1)}| \leq |A_{\phi}^{(1)}|/2$, we have that 
$\prob{z_1, \ldots, z_t \in S} \leq (1/2)^t$. Thus, the required probability 
bounded  above by
\[\sum_{|S \cap A_{\phi}^{(1)}| \leq |A_{\phi}^{(1)}|/2} (1/2)^t \cdot 1 
~~\leq~ 2^{-t}  O(2^{4n^2}) .\]
The last bound uses the fact that the number of subspaces
of $\F_2^{2n}$ is $O(2^{4n^2})$. Thus, for $t = 4n^2 + \log(10/\delta)$, the
probability is at most $\delta/10$.

We now bound the probability that the sampled points $z_1, \ldots, z_t$ span a subspace of 
dimension at most $n-k$. The probability that a random of $A_{\phi}^{(1)}$ lies in 
a \emph{specific} subspace of dimension $n-k$ is at most $(2^{-k}/(\rho/6))$. Hence, the
probability that all $t$ points lie in any subspace of dimension $n-k$ is bounded above by
\[ \inparen{\frac{2^{-k}}{\rho/6}}^t \cdot \#\{\text{subspaces of dim}~  n-k\} 
~\leq~ \inparen{\frac{2^{-k}}{\rho/6}}^t \cdot 2^{n(n-k)}  .\]
%
For $t \geq n^2 + O(\log(1/\delta))$ and $k = \log(12/\rho)$, this probability is at 
most $\delta/10$. Hence the dimension of the span of the sampled vectors is at least 
$n - \log(12/\rho)$ with high probability. 
\end{proof}

Next, we \emph{upper bound} the dimension of the span of the retained points (on which 
\bsgtest answered 1). 
By the assumed correctness of \bsgtest, we 
get that all the points must lie inside $A_{\phi}^{(2)}$. 
Applying the Freiman-Ruzsa Theorem (Theorem \ref{thm:freiman}), it follows that
\[ |< A_\phi^{(2)}>| ~\leq~ \exp(1/\rho^C) N. \]
The above implies that all the points are inside a space of dimension
at most $n + \log(1/\nu)$, where we have written $\nu = \exp(-1/\rho^C)$. From here, we can proceed in a similar fashion to \cite{Samorodnitsky07}.

Let $V$ denote the span of the retained points and let $v_1, \ldots, v_{r}$ be a basis for $V$.
We can add vectors to complete it to $v_1, \ldots, v_{s}$ so that the projection onto the first $n$
coordinates has full rank. Let $V' = <v_1, \ldots, v_s>$.
We can also assume, by a change of basis, that for $i \leq n$ we have the coordinate vectors
$v_i = (e_i, u_i)$. This can
all be implemented by performing Gaussian elimination, which takes time $O(n^3)$.

Consider the $2n \times s$ matrix with $v_1, \ldots, v_s$ as columns. By the 
previous discussion, this matrix is of the form 
\[P = \left( \begin{array}{cc} I & 0 \\ T & U \end{array}\right),\]
where $I$ is the $n \times n$ identity matrix, and $T$ and $U$ are $n \times n$ and 
$n \times (s-n)$ matrices, respectively. By Claim \ref{clm:large-span}, we know that
$v'$ contains $|A_{\phi}^{(1)}|/2 \geq (\rho/12) N$ vectors of the form $(x,\phi(x))^T$.
For each such vector, there exists a $w \in \F_2^{s}$ such that
$P \cdot w = (x, \phi(x))^T$. Because of the form of $P$, we must have that $w = (x,z)$
for $z \in \F_2^{s-n}$. Thus, we get that for each vector $(x,\phi(x))$, we in fact have $\phi(x) = Tx + Uz$
for some $z \in \F_2^{s-n}$.

Therefore, for at least one $z_0 \in \F_2^{s-n}$ and $y_0 = Uz_0$ we find that
\[ \Prob{x \in \F_2^n}{\phi(x) = Tx + y_0} ~\geq~ (\rho/12) \cdot 2^{-(s-n)}.\]
We next upper bound $s-n$. Note that $s \leq r + k$ since by Claim \ref{clm:large-span}, $V$ had 
dimension at least $n-k$ for $k = \log(12/\rho)$. Also, we know 
that $r \leq n + \log(1/\nu)$  by the bound on $|< A_{\phi}^{(2)}>|$, implying that 
$s \leq n + \log(12/\rho) + \log(1/\nu)$. We conclude that  $2^{-(s-n)} \geq (\rho/12) \nu$.

Moreover, for each element of the form $(x,\phi(x)) \in A_{\phi}^{(1)}$, we know that
$\smallabs{\fxhat{(\phi(x))}} \geq \gamma \geq \e^{16}/180$. This implies that
\[ \Ex{x \in \F_2^n}{\fxhat^2(Tx + y_0)} \geq \gamma^2 \cdot (\rho/12) \cdot (\rho \nu/12) .\]
Samorodnitsky shows that we can in fact take $y_0$ to be 0. In fact, he shows the following
general claim.
\begin{claim}[Consequence of Lemma 6.10 \cite{Samorodnitsky07}]
For any matrix $T$ and $y \in \F_2^n$, 
$\Ex{x \in \F_2^n}{\fxhat^2(Tx + y)} \leq \Ex{x \in \F_2^n}{\fxhat^2(Tx)}$.
\end{claim}
Thus, we simply output the matrix $T$ constructed as above. For $\eta = \gamma^2 \rho^2 \nu /144$,
it satisfies $\Ex{x \in \F_2^n}{\fxhat^2(Tx)} \geq \eta$. Finally, we calculate the probability that
the algorithm outputs $\bot$ or outputs a $T$ not satisfying this guarantee. This can happen only
when the guarantee on \bsgtest is not satisfied for one of the sampled points, or when the
guarantees in Claims \ref{clm:unisamples} and \ref{clm:large-span} are not satisfied. Since each of
these happen with probability at most $\delta/5$, the probability of error is at most $3\delta/5 <
\delta$.
\end{proof}

\subsection{Finding a quadratic phase function} \label{sec:quadratic-phase}
Once we have identified the linear map $T$ above, the remaining argument is identical to the one in 
\cite{Samorodnitsky07}. 

Equipped with $T$, one can find a symmetric matrix $B$ with zero diagonal
that satisfies a slightly weaker guarantee. This step is usually referred to as the \emph{symmetry argument}, and we shall encounter a modification of it in Section \ref{sec:refinement}. The only algorithmic steps
used in the process are Gaussian elimination and finding a basis for a subspace, which
can both be done in time $O(n^3)$.


\begin{lemma}[Proof of Theorem 2.3 \cite{Samorodnitsky07}]\label{lem:symmetrization}
Let $T$ be as above. Then in time $O(n^3)$ one can find a symmetric matrix $B$ with zero 
diagonal such that $\Ex{x \in \F_2^n}{\fxhat^2(Bx)} \geq \eta^2$.
\end{lemma}

Now that we have correlation of the derivative $f_x$ of the function with a truly linear map, it remains to ``integrate" this relationship to obtain that $f$ itself correlates with a quadratic map. Following Green and Tao, we shall henceforth refer to this part of the argument as the \emph{integration step}.

Having obtained $B$ above, we can find a matrix $M$ such that $M + M^T = B$.  We take the
quadratic part of the phase function to be $h(x) = (-1)^{\langle x,M x\rangle}$. The following claim
helps establish the linear part.


\begin{lemma}[Corollary 6.4 \cite{Samorodnitsky07}]\label{lem:integration}
Let $B$ and $h$ be as above. Then there exists $\alpha \in \F_2^n$ such that 
$\smallabs{\widehat{fh}(\alpha)} \geq \eta^2$.
\end{lemma}

An appropriate $\alpha$ can be found using the algorithm \LinearDecomposition with parameter 
$\gamma' = \eta^2$ (by picking any element from the list it outputs). 
We take $q(x) = {\langle x,M x\rangle + \langle\alpha, x\rangle + c}$ where $(-1)^{c}$ is the sign of the coefficient for
$(-1)^{\langle \alpha,  x\rangle }$ in the linear decomposition. The running time of this 
step is $O(n^3 \log n \cdot \poly(1/\eta, \log(1/\delta)))$, where $\delta$ is the probability of
error we want to allow for this invocation of \LinearDecomposition. 

Note that of all the steps involved
in finding a quadratic phase, finding the \emph{linear} part of the phase is the only step for which
running time depends exponentially on $\e$ (since $\eta = \exp(-1/\e^{\Omega(1)})$). The running
time of all other steps depends polynomially on $1/\e$.

\subsection{Putting things together} \label{sec:final-quadratic}
We are now ready to finish the proof of Theorem \ref{thm:FindQuadratic}.
\begin{proofof}{of Theorem \ref{thm:FindQuadratic}}
For the procedure \FindQuadratic the function $\phi(x)$ will be sampled using Lemma 
\ref{lem:sample-phi} as required. We start with a random $u = (x,\phi(x))$
and a random choice for the parameters $\gamma_1,\gamma_2,\gamma_3$ as
described in the analysis of \bsgtest. We run the algorithm in Lemma \ref{lem:lin-choice} 
using \bsgtest with the above parameters and with error parameter $1/2$.

If the algorithm outputs  a quadratic form $q(x)$,
we estimate $\abs{\ip{f,(-1)^q}}$ using $O((1/\eta^4) \cdot \log^2(\rho/\delta))$
samples. If the estimate is less than $\eta^2/2$, or if the algorithm stopped with output 
$\bot$ we discard $q$ and repeat
the entire process. 
For a $M$ to be chosen later, if we do not find
a quadratic phase in $M$ attempts, we stop and output $\bot$.

With probability $\rho/2$, all samples of $\phi(x)$ (sampled with error $1/n^5$)
correspond to a good function $\phi$. Conditioned on this, we have a good choice of
$u$ and $\gamma_1,\gamma_2,\gamma_3$ for \bsgtest with probability $\rho^3/24$. 
Conditioned on both the above, 
the algorithm in Lemma \ref{lem:lin-choice} finds a good transformation with probability
$1/2$. Thus, for $M = O((1/\rho^4 ) \cdot \log(1/\delta))$, the algorithm stops in 
$M$ attempts with probability at least $1-\delta/2$. By choice of the number of
samples above, the probability that we estimate $\abs{\ip{f,(-1)^q}}$ incorrectly at any step is at most 
$\delta/2M$. Thus, with probability at least $1-\delta$, we output a good quadratic phase.

One call to the algorithm in Lemma \ref{lem:lin-choice} requires $O(n^2)$ calls to 
\bsgtest, which in turn requires
$\poly(1/\e)$ calls to \LinearDecomposition, each taking time $O(n^2 \log n)$. This
dominates the running time of the algorithm, which is 
$O(n^4 \log n \cdot \poly(1/\e, 1/\eta, \log(1/\delta)))$. 
\end{proofof}


\section{A refinement of the inverse theorem}\label{sec:refinement}

In this section we shall work with a number of refinements of the inverse theorem as stated in
Theorem \ref{thm:globalinverse}. For the purposes of the preliminary discussion we shall think of
$p$ being any prime, and later specialize to the case $p=2$.

It was observed (but not exploited) by Green and Tao \cite{GrTu3} that a slightly stronger form of
the inverse theorem holds. If $V$ is a subspace of $\F_p^n$ and $y\in\F_p^n$, then one can define a
seminorm $\|.\|_{u^3(y+V)}$ on functions from $\F_p^n$ to $\C$ by setting 
$$\|f\|_{u^3(y+V)}=\sup_q |\E_{x\in y+V}f(x)\omega^{-q(x)}|,$$ 
where the supremum is taken over all quadratic forms $q$ on $y+V$ and $\omega$ denotes a $p$th root
of unity. This semi-norm measures the correlation over a coset of the subspace $V$. We shall
be interested in the co-dimension of the subspace, which we shall denote by $\cod V$.
With this notation, the inverse theorem in \cite{GrTu3} can be stated as follows.


\begin{theorem}[Local Inverse Theorem for $U^3$ \cite{GrTu3}] \label{thm:localinverse}
Let $p>2$, and let $f:\F_p^n\ra\C$ be a function such that
$\|f\|_\infty\leq 1$ and $\|f\|_{U^3}\geq\eps$. 
Then there exists a subspace $V$ of $\F_p^n$ such that $\cod V \leq \eps^{-C}$ and
\[\E_{y \in \comp{V}}\|f\|_{u^3(y+V)}\geq\eps^{C}.\]
\end{theorem}

Here we have denoted the set of coset representatives of $V$ by $\comp{V}$, so that $V \oplus\comp{V}=\F_2^n$. Actually, the theorem as usually stated involves an averages over the whole of $\F_p^n$ as opposed
to just $\comp{V}$, but the result can be obtained with this modification without difficulty by
averaging over coset representatives throughout the proof.

One can deduce the usual inverse theorem from this version without too much effort: by an averaging
argument, there must exist
$y$ such that $f$ correlates well on $y+V$ with some quadratic
phase function $\omega^q$; this function can be extended to a 
function on the whole of $\F_p^n$ in many different ways, and
a further averaging argument yields the usual bounds. However, extending the quadratic phase results
in an exponential loss in correlation. (See, for example, Proposition 3.2 in \cite{GrTu3}.)

It turns out that, as Green and Tao remark, an even more precise
theorem holds. The result as stated tells us that for each $y$ we
can find a local quadratic phase function $\omega^{q_y}$ defined on $y+V$
such that the average of $|\E_{x\in y+V}f(x)\omega^{q_y(x)}|$ is at
least $\eps^{C}$. However, it is actually possible to do this in 
such a way that the quadratic parts of the quadratic phase 
functions $q_y$ are the same. More precisely, it can be done in 
such a way that each $q_y(x)$ has the form $q(x-y)+l_y(x-y)$ for
a single quadratic function $q:V\ra\F_p$ (that is independent of $y$) and 
some Freiman 2-homomorphisms $l_y:V\ra\F_p$.

This \emph{parallel correlation} was heavily exploited by Gowers and the second author
\cite{GW2,GW4} in a series of papers on what they called the \emph{true complexity} of a system of
linear equations, leading to radically improved bounds compared with the original approach in
\cite{GW1}, which was based on an ergodic-style decomposition theorem due to Green and Tao
\cite{GrML}.

For $p=2$, the equivalent of Theorem \ref{thm:localinverse} follows directly neither from Green and Tao's nor
Samorodnitsky's approach but instead requires a merging of the two. The Green-Tao approach is not
directly applicable since the so-called \emph{symmetry argument} in that paper uses division by 2,
while Samorodnitsky's approach loses the local information after an application of Freiman's
theorem.
Section \ref{sec:refinement} is dedicated to showing how to obtain this \emph{local
  correlation}\footnote{The term ``local correlation" may be slightly confusing. It is often used to
  refer to the fact that in $\Z/N\Z$, no global quadratic correlation with a quadratic phase can be
  guaranteed. Indeed, such a phase function must be restricted to a Bohr set, or the correlation
  assumed to only take place on a long arithmetic progression, as in Gowers's original
  work. However, in $\F_p^n$, the setting we are working in here, there should be no ambiguity.} in
the case where the characteristic is equal to 2. We shall therefore restrict our attention to this
case for the remainder of the discussion, bearing in mind that it applies almost verbatim to 
general $p$.

In order to be able to refer to the parallel correlation property more concisely, we shall use the
concept of \emph{quadratic averages} introduced in \cite{GW2}. As explained above, for each coset
$y+V, y \in \comp{V}$, we can specify a quadratic phase $q_y(x)=q(x-y)+l_y(x-y)$. We extend the
definition of $q_y$ to all $y \in \F_p^n$ by setting them equal to $q_{\hat{y}}$ where $\hat{y} \in
\comp{V}$ is such that $y \in \hat{y}+V$.
Now we can define a quadratic average via the formula
\[Q(x)=\E_{y\in x-V}(-1)^{q_y(x)}.\]
Notice that the $q_y$ are the same whenever the $y$ lie in the same coset of $V$. So in fact, since
all the $q_y$s occurring here are such that $y \in x+V$, they are all identical. Thus the value of
the quadratic average only depends on the coset of $V$ that $x$ lies in. More precisely, we can
write
\[Q(x)=\sum_{y \in \comp{V}}1_{y+V}(x)(-1)^{q_y(x)}.\]
This tells us that at most $|\comp{V}|$ many linear phases are needed to specify the quadratic
average.

Combining the Green-Tao approach with Samorodnitsky's symmetry argument in characteristic 2, we shall obtain an algorithmic version of the analogue of the Local Inverse Theorem
(Theorem \ref{thm:localinverse}) for $p=2$.
In order to use this result in our decomposition algorithm Theorem
\ref{thm:decomposition-general}, we in fact state it as an algorithm for finding a \emph{quadratic average}
$Q(x)=\sum_{y \in \comp{V}}1_{y+V}(x)(-1)^{q_y(x)}$, 
which has correlation $\poly(\e)$ with the given function. Using this, Theorem \ref{thm:decomposition-general} will then yield a decomposition into
$\poly(1/\e)$ quadratic averages.

Following \cite{GW1}, we shall call the codimension of $V$ the \emph{complexity} of the 
quadratic average. We will find quadratic averages with complexity $\poly(1/\e)$. 
Note that while this means that the description of a quadratic average is still of size $\exp(1/\e)$, the different
quadratic forms appearing in a quadratic average only differ in the linear part.

\begin{theorem}\label{thm:FindQuadraticAverage}
Given $\e, \delta > 0$ and $n \in \N$, there exist $K,C = O(1)$ and
a randomized algorithm \FindQuadraticAverage running in time 
$O(n^4 \log^2 n \cdot \exp(1/\e^K) \cdot \log(1/\delta))
$, which,  given oracle access to a function $f: \F_2^n \to \pmone$, either outputs a quadratic average $Q(x)$
of complexity $O(\e^{-C})$, or the symbol $\bot$. The algorithm satisfies the following guarantee:
\begin{itemize}
\item If $\uthreenorm{f} \geq \e$, then with probability at least $1-\delta$  it finds a quadratic
  average $Q$ of complexity $O(\e^{-C})$ such that $\ip{f,Q} \geq \eps^{C}$.
\item The probability that the algorithm outputs a $Q$ which has $\ip{f,Q} \leq \eps^{C}/2$ is at 
most $\delta$.
\end{itemize}
\end{theorem}

We briefly outline the key modifications in the proof that allow us to obtain this result. Recall that in the previous section we only obtained correlation $\eta =
\exp(1/\e^C)$ because we applied the Freiman-Ruzsa theorem to the set
$A_{\phi}^{(2)}$: we were only able to assert that 
$|< A_{\phi}^{(2)}>| \leq \exp(1/\e^C) |A_{\phi}^{(2)}|$. Because we had
correlation $\poly(\e)$ over $A_{\phi}^{(2)}$, we obtained correlation $\exp(-1/\e^C)$ with the
linear function we defined on $< A_{\phi}^{(2)} >$. 

They key difference in the new argument, which borrows heavily from Green and Tao \cite{GrTu3}, is that instead of looking for a subspace \emph{containing}
$A_{\phi}^{(2)}$, which we previously used to find a linear function, we will look for a subspace
\emph{inside} $4A_{\phi}^{(2)}$. Given the properties of $A_{\phi}^{(2)}$, we will be able to find
such a subspace by an application of Bogolyubov's lemma (described in more detail below), with the property that the
co-dimension of the subspace is $\poly(1/\e)$. We will also find a quadratic form such that
\emph{restricted to inputs from this subspace}, it has correlation $\poly(1/\e)$ with the function
$f$. We shall then show (Lemma \ref{lem:FCtoCorr2}) how to extend this quadratic form to all the
cosets of the subspace, by adding a \emph{different linear form for each coset} so that the
correlation of the resulting quadratic average is still $\poly(1/\e)$.

We begin by developing algorithmic version of some of the new ingredients in the proof.

\subsection{An algorithmic version of Bogolyubov's lemma} \label{sec:bogolyubov}
We follow Green and Tao in using a form of Bogolyubov's lemma, which has become a
standard tool in arithmetic combinatorics. Bogolyubov's lemma as it is usually stated allows one to
find a large subspace inside the 4-fold sumset of any given set of large size.
We briefly remind the reader of the  relationship between sumsets and convolutions, which is used in
the proof of the lemma.

For functions $h_1, h_2: \F_2^n \to \R$, we define their convolution as
$h_1 * h_2 (x) ~\defeq~ \Ex{y}{h_1 (y) h_2 (x-y)}$.
The Fourier transform diagonalizes the convolution operator, that is,
$\widehat{h_1 * h_2}(\alpha) = \widehat{h_1}(\alpha)\widehat{h_2}(\alpha)$ for any two functions
$h_1,h_2$ and any $\alpha \in \F_2^n$ , which is easy to verify
from the definition.
Also, if $1_A$ is the indicator function for a set $A \subseteq \F_2^n$, then
\[ 1_A * 1_A (x) ~=~ \Ex{y}{1_A (y) \cdot 1_A (x-y)} 
~=~ \abs{\{ (y_1,y_2)  \suchthat y_1,y_2 \in A ~\text{and}~ y_1 + y_2 = x\}}/2^n . \]
In particular, $1_A * 1_A$ is supported only on $A+A$ and gives the number of representations of $x$ as the sum of two elements in $A$. In general, the $k$-fold convolution is
supported on the $k$-fold sumset. 

The proof of Bogolyubov's lemma constructs an explicit
subspace by looking at the large Fourier coefficients (using the Goldreich-Levin theorem) and shows
that the 4-fold convolution is positive on this subspace. Since we will actually apply this lemma
not to a subset but to the output of a randomized algorithm, we state it for an arbitrary function
$h$ and its convolution. 

We will output a subspace $V \subseteq \F_2^n$ by specifying a basis for
the space $\ortho{V} \defeq \{ x : x^T y = 0 ~~\forall y \in V\} $. Since $\ortho{(\ortho{V})} = V$,
this will also give us a way of checking if $x \in V$: we simply test if $x^T y = 0$ for all basis
vectors $y$ of $\ortho{V}$.

\begin{lemma}[Bogolyubov's Lemma]\label{lem:bogolyubov}
There exists a randomized algorithm \bogolyubov with parameters $\rho$ and $\delta$ which, given
oracle access to a function $h: \F_2^n \ra \{0,1\}$ with $\E h \geq \rho$, outputs a subspace $V 
\leqslant \F_{2}^{n}$ (by giving a basis for $\ortho{V}$) of codimension at most $O(\rho^{-3})$ such
that with probability at least $1-\delta$, we have $h*h*h*h(x)>\rho^4/2$ for all  $x \in V$. The
algorithm runs in time $n^2 \log n\cdot \poly(1/\rho, \log(1/\d))$.
\end{lemma} 
\begin{proof}
We shall use the Goldreich-Levin algorithm \LinearDecomposition for the function $h$ with parameter
$\gamma=\rho^{3/2}/4$ and error $\delta$ to produce a list $K=\{\alpha_1, \dots, \alpha_k\}$ of
length $k=O(\gamma^{-2})=O(\rho^{-3})$.
We take $V$ to be the subspace 
$\{ x \in \F_2^n \suchthat \langle \alpha, x \rangle= 0 ~~\forall \alpha \in K\}$
and output $\langle K \rangle$. 
Clearly $\cod(V)\leq|K|$. We next consider the convolution
\[h*h*h*h(x)=\sum_{\alpha }|\widehat{h}(\alpha )|^{4}(-1)^{\langle \alpha, x\rangle}=\sum_{\alpha  \in
  K}|\widehat{h}(\alpha)|^{4}(-1)^{\langle\alpha, x\rangle}+\sum_{\alpha \not\in
  K}|\widehat{h}(\alpha)|^{4}(-1)^{\langle\alpha, x\rangle}.\]
If $x \in V$, then
\[\sum_{\alpha \in K}|\widehat{h}(\alpha)|^{4}(-1)^{\langle \alpha, x\rangle}+\sum_{\alpha \not\in
  K}|\widehat{h}(\alpha)|^{4}(-1)^{\langle \alpha, x\rangle} \geq |\widehat{h}(0)|^{4}-\sup_{\alpha \notin K}
|\widehat{h}(\alpha)|^{2}\cdot \rho\]
The final part of the guarantee in Theorem \ref{thm:linear-decomposition} states that the
probability of a Fourier coefficient being larger than $\gamma$ and not being on our list $K$ is at
most $\delta$. We conclude that with probability at least $1-\delta$, the expression $h*h*h*h(x)$ is
bounded below, for all $x \in V$, by
\[\rho^4-\rho\cdot\rho^{3}/2  \geq \rho^{4}/2,\]
and thus strictly positive.
\end{proof}

We will, in fact, need a further twist of the above lemma. The function $h$ to which will apply
Lemma \ref{lem:bogolyubov} will be defined by the output of a randomized algorithm. Thus, $h$ can be
thought of as a random variable, where we choose the value $h(x)$ on each input $x$ by running the
randomized algorithm. As in the case of \bsgtest, we will have the guarantee that there exist two sets
$A^{(1)} \subseteq A^{(2)}$ and $\delta' > 0$ such that \emph{for each input $x$}, with probability
  $1-\delta'$ (over the choice of $h(x)$) we have $1_{A^{(1)}}(x) \leq h(x) \leq 1_{A^{(2)}}(x)$. We
  will want to use this to conclude that \emph{for the entire subspace $V$} given by the algorithm
\bogolyubov, $V \subseteq 4A^{(2)}$. 

To argue this, it will be useful to consider the function $h'$ defined as 
$h' \defeq \min\{ 1_{A^{(2)}}, \max\{h, 1_{A^{(1)}}\} \}$. By definition, we always have that
$1_{A^{(1)}}(x) \leq h'(x) \leq 1_{A^{(2)}}(x)$. Also, if for each $x$, we have with probability $1-\delta'$
$1_{A^{(1)}}(x) \leq h(x) \leq 1_{A^{(2)}}(x)$, this means that
for each $x$,  $\prob{h(x) \neq h'(x)} \leq \delta'$.
The following claim gives the desired conclusion for the subspace given by the algorithm \bogolyubov.
\begin{claim} \label{clm:noisy-bogolyubov}
Let $h$ be a random function such that for $\delta' > 0$ and for sets $A^{(1)} \subseteq A^{(2)}
\subseteq \F_2^n$, we have that for every $x$ with probability at least $1-\delta'$, 
$1_{A_{(1)}}(x) \leq h(x) \leq 1_{A^{(2)}}(x)$. Also, let $\E 1_{A^{(1)}} \geq \rho$.
Let $h' = \min\{ 1_{A^{(2)}}, \max\{h, 1_{A^{(1)}}\} \}$
Let $V$ be the subspace returned by the algorithm
\bogolyubov when run with oracle access to $h$ and error parameter $\delta$. Then with probability
at least $1-\delta - \delta' \cdot n^2 \log n \cdot \poly(1/\rho, \log(1/\delta))$, we have that 
for all $x \in V$, $1_{A^{(2)}} * 1_{A^{(2)}} * 1_{A^{(2)}} * 1_{A^{(2)}} (x)\geq h' * h' * h' * h'
(x)> \rho^4/2$. In particular, with above probability, $V \subseteq 4A^{(2)}$.
\end{claim}
\begin{proof}
Consider the behavior of the algorithm \bogolyubov when run with oracle access to $h'$ instead of
$h$. Since it is always true that $h' \leq 1_{A^{(2)}}$ and $\ex{h'} \geq \ex{1_{A^{(1)}}} \geq
\rho$, the algorithm outputs, with
probability $1-\delta$, a subspace $V$ such that for every $x \in V$, 
$1_{A^{(2)}} * 1_{A^{(2)}} * 1_{A^{(2)}} * 1_{A^{(2)}} (x)\geq h' * h' * h' * h' (x)> \rho^4/2$. Thus, with
probability $1-\delta$, it outputs a subspace $V$ such that $V \subseteq 4 A^{(2)}$.

Finally, we observe that the probability that the algorithm outputs different subspaces when run
with oracle access to $h$ and $h'$ is small. The probability of having different outputs is at most
the probability that $h$ and $h'$ differ on any of inputs queried by the algorithm
\bogolyubov. Since it runs in time $n^2 \log n\cdot \poly(1/\rho, \log(1/\d))$, this probability is
at most $\delta' \cdot n^2 \log n\cdot \poly(1/\rho, \log(1/\d))$. Thus, even when run with oracle
access to $h$, with probability at least $1 - \delta - \delta' \cdot n^2 \log n\cdot
\poly(1/\rho, \log(1/\d))$, the algorithm \bogolyubov outputs a subspace $V \subseteq 4A^{(2)}$. 
\end{proof}
Next we require a version of Pl\"unnecke's inequality in order to deal with the size of iterated
sumsets. For a proof we refer the interested reader to \cite{TaoVu}, or the recent short and elegant
proof by Petridis \cite{Petridis}.


\begin{lemma}[\Plunnecke's Inequality]\label{lem:plunnecke}
Let $B \subseteq \F_{2}^{n}$ be such that $|B+B| \leq K|B|$ for some $K>1$. Then for any positive
integer $k$, we have $|kB|\leq K^{k}|B|$.
\end{lemma}

\subsection{Finding a good model set} \label{sec:model}

Again, as in Section \ref{sec:correlations} we may assume that $\phi$ is a good function satisfying
the guarantee in Lemma \ref{lem:sample-phi}.
Recall that $A_\phi=\{(x,\phi(x)): x \in A\}$, where $A$ was defined to be 
$A=\{x: |\widehat{f_x}(\phi(x))| \geq \gamma\}$. 
We will use the routine \bsgtest described in Section \ref{sec:correlations}.
We assume we have chosen a good vertex $u$ 
and parameters $\gamma_1,\gamma_2,\gamma_3$ satisfying the guarantee in Lemma \ref{lem:bsgtest} for
\bsgtest.

We will need to restrict the sets $A_{\phi}^{(1)}$ and $A_{\phi}^{(2)}$  given by Lemma
\ref{lem:bsgtest} a bit more before we can
apply Bogolyubov's lemma to find an appropriate subspace. Because the subspace sits inside the sumset
$4 A_{\phi}^{(2)}$, an element of the subspace is of the form 
$(x_1 + x_2 + x_3 + x_4, \phi(x_1) + \phi(x_2) + \phi(x_3) + \phi(x_4))$. However, unlike tuples of the form $(x,\phi(x))$, the second half of the tuple ($\phi(x_1) +
\phi(x_2) + \phi(x_3) + \phi(x_4)$) may not uniquely depend on the first ($x_1 + x_2 + x_3 +
x_4$). 

Since we will require this uniqueness property from our subspace, we restrict our sets to get new sets  
$A_{\phi}'^{(1)} \subseteq A_{\phi}'^{(2)}$. These restrictions will satisfy the following property:
for all tuples $x_1, x_2, x_3, x_4$ and $x_1', x_2', x_3', x_4'$ satisfying $x_1 +  x_2  +  x_3  +
x_4 = x_1'+ x_2'+ x_3'+ x_4'$, we also have $\phi(x_1) + \phi(x_2) + \phi(x_3) + \phi(x_4) = \phi(x_1') +
\phi(x_2') + \phi(x_3') + \phi(x_4)'$. In other words, $\phi$ is a Freiman 4-homomorphism on the
first $n$ coordinates of $A_{\phi}'^{(2)}$.
We will, in fact, need to ensure that it is a Freiman 8-homomorphism in order to obtain a truly linear map.

We shall obtain these restrictions by intersecting the original sets with a subspace, which will be defined
using a random linear map $\Gamma: \F_2^n \ra \F_2^m$ and a random element $c \in \F_2^m$ 
(for $m = O(\log(1/\e))$). This step is often called \emph{finding a good model}, and appears (in
non-algorithmic form) as Lemma 6.2 in \cite{GrTu3}.
We shall apply the restriction $\Gamma(\phi(x)) = c$ to the elements $v = (x,\phi(x))$ on which
\bsgtest outputs 1. Since we assume we have already chosen good parameters 
$u, \rho_1,\rho_2,\gamma_1, \gamma_2, \gamma_3$ for the routine
\bsgtest, we hide these parameters in the description of the procedure below.

\medskip

\fbox{
\begin{minipage}{0.9\textwidth}

\smallskip

\ModelTest(v, $\Gamma$, c)
\begin{itemize}
\item[-] Let $v=(y,\phi(y))$. 
\item[-] Answer 1 if \bsgtest returns 1  on $v$ and $\Gamma(\phi(y))=c$, and 0 otherwise. 
\end{itemize}

\end{minipage}
}

\medskip

We shall first show that there exist \emph{good} choices of $\Gamma$ and $c$ for our purposes. Let
$A_\phi^{(2)}$ be the set provided by Lemma \ref{lem:bsgtest} for a good choice of parameters. Let
$B \subseteq \F_2^n\setminus\{0\}$ be the set of all $t$ such that $(0,t) \in 16A_\phi^{(2)}$.


\begin{claim}\label{clm:sizeb}
Let $\theta'=\eps^{2448}/2^{487}$. The set $B$ has size at most $\theta'^{-1}$.
\end{claim}
\begin{proof}
Write $(0,B)$ for the set of all $(0,b), b\in B$. Since $A_\phi^{(2)}$ is of the form $(x,\phi(x))$ for some function $\phi$, we have $|A_\phi^{(2)}+(0,B)|=|A_\phi^{(2)}||B|$, but at the same time $A_\phi^{(2)}+(0,B) \subseteq 17A_\phi^{(2)}$. By Lemma \ref{lem:plunnecke} we have $|17A_\phi^{(2)}| \leq (3(2/\rho)^9)^{17} |A_\phi^{(2)}| \leq (2^{181}/\rho^{153}) |A_\phi^{(2)}|$ since $A_\phi^{(2)}$ has small sumset, and therefore $|B| \leq 2^{181}/\rho^{153}=\theta'^{-1}$, since $\rho=\eps^{16}/4$.
\end{proof}


\begin{claim}\label{clm:probb}
Let $m=2\lceil\log_2 \theta'^{-1}\rceil$. Then with probability at least 1/2 a random linear map
$\Gamma: \F_2^n \ra \F_2^m$ is non-zero on all of $B$.
\end{claim}
\begin{proof}
Let $\Gamma: \F_2^n \ra \F_2^m$ be a randomly chosen linear transformation. Let $E_t$ be the event
that $\Gamma(t)=0$. Clearly $\Psymb(E_t)\leq 2^{-m}$ for each $t \in B$, and thus the probability
that $\Gamma$ is non-zero on all of $B$ is $\Psymb(\cap_t (E_t^C))=\Psymb((\cup_t
E_t)^C)=1-\Psymb(\cup_t E_t)\geq 1-\sum_t \Psymb(E_t) \geq 1-|B|2^{-m} \geq 1/2$ by choice of
$m$. So with probability at least $1/2$ we have a map $\Gamma$ that is non-zero on $B$. 
\end{proof}


\begin{claim}\label{clm:subsetsize}
Let $\theta=\theta'^2 \rho/12$, where $\theta'$ is the constant obtained in Claim \ref{clm:sizeb}, that is, we set $\theta=\eps^{4912}/(3\cdot 2^{977})$. Fix a map $\Gamma$ as in Claim \ref{clm:probb}. Then with probability at least $\theta$ a randomly chosen element $c \in\F_2^m$ is such that the set
\[A_\phi'^{(1)}\defeq\{(x,\phi(x)) \in A_\phi^{(1)}: \Gamma(\phi(x)) =c \}\]
has size at least $\theta N$.
\end{claim}
\begin{proof}
The expected size of this set is at least $|A_\phi^{(1)}|/2^m \geq (\rho N/6)/ (\theta'^{-2}) \geq (\theta'^{2} \rho/6) N$, so with probability $\theta$ we can get it to be of size at least $\theta N$.
\end{proof}

We shall of course also define
\[A_\phi'^{(2)}\defeq\{(x,\phi(x)) \in A_\phi^{(2)}: \Gamma(\phi(x)) =c \},\]
and since $ A_\phi^{(1)}\subseteq  A_\phi^{(2)}$, we have a similar containment for the new subsets,
immediately giving a similar lower bound on the size of $A_\phi'^{(2)}$.

We summarize the above claims in the following refinement of Lemma \ref{lem:bsgtest}.


\begin{lemma}\label{lem:modeltest}
Let the calls to \bsgtest in \ModelTest be with a good choice of parameters
$u,\rho_1,\rho_2,\gamma_1,\gamma_2,\gamma_3$ and with error parameter $\delta > 0$. Then, there
exist two sets  $A_{\phi}'^{(1)} \subseteq A_{\phi}'^{(2)}$,  the output of \ModelTest on input 
$v = (y,\phi(y))$ satisfies the
following with probability $1-\delta$.
\begin{itemize}
\item $\ModelTest(v,\Gamma,c) = 1 
\quad\Longrightarrow\quad v \in A_{\phi}'^{(2)}$.
\item $\ModelTest(v,\Gamma,c) = 0 
\quad\Longrightarrow\quad v \notin A_{\phi}'^{(1)}$.
\end{itemize}
Moreover, with probability $\theta/2$ over the choice of $\Gamma$ and $c$ , we have
\[ |A_{\phi}'^{(1)}| \geq \theta N \quad \text{and} 
\quad \phi \text{ is a Freiman 8-homomorphism on }   A^{(2)},\quad\]
where we denote the projection of $A_\phi'^{(2)}$ onto the first $n$ coordinates by $A^{(2)}$.
\end{lemma}
\begin{proof}
If \ModelTest outputs 1, then $v=(y, \phi(y)) \in A_\phi^{(2)}$ with probability $1-\delta$ and
$\Gamma(\phi(y)) =c$, so $v \in A_\phi'^{(2)}$. Similarly, if \ModelTest outputs 0 then either
\bsgtest gave 0 or $\Gamma(\phi(y)) \neq c$, so in any case $v \not\in A_\phi'^{(1)}$.

By Claims \ref{clm:subsetsize} and \ref{clm:probb}, with probability at least $\theta/2$ over the
choice of $\Gamma$ and $c$, $|A_{\phi}'^{(1)}| \geq \theta N$ and $\Gamma$ is non-zero on all of
$B$.  It remains to verify that $\phi$ is a Freiman 8-homomorphism on $A^{(2)}$ in this case. 

For any $(0,t) \in 16A_\phi'^{(2)}$, we have $t \neq 0 \implies t \in B$ by definition.
Also $\Gamma(t)=16 c = 0$ by linearity of $\Gamma$. Since $\Gamma$ is non-zero on all of $B$, we
must have $t=0$. 
We also have $16A_\phi'^{(2)}=8A_\phi'^{(2)}+8A_\phi'^{(2)}$, and so if we take
$(0,t)=(x_1+\dots+x_8+x_1'+ \dots x_8',\phi(x_1)+\dots+\phi(x_8)+\phi(x_1')+ \dots \phi(x_8'))$, we
have that $x_1+\dots+x_8+x_1'+ \dots x_8'=0$ implies $\phi(x_1)+\dots+\phi(x_8)+\phi(x_1')+ \dots
\phi(x_8')=0$, making $\phi$ a Freiman 8-homomorphism on $A^{(2)}$.
\end{proof}

\subsection{Obtaining a linear choice function on a subspace} \label{sec:lin-locchoice}

As before, we now identify a linear transform (actually, an affine transform) that selects large
Fourier coefficients in derivatives. However, as opposed to Section \ref{sec:correlations} where we
defined a linear transform on the whole of $\F_2^n$, here we will just define it on a
\emph{coset a subspace $V$} such that $\cod(V) = \poly(1/\e)$. 

In particular, we will prove the following local version of Lemma \ref{lem:lin-choice}.


\begin{lemma}\label{lem:lin-locchoice}
Let $\phi$ be as above and let the parameters for \bsgtest and \ModelTest be so that they satisfy
the guarantees of lemmas \ref{lem:bsgtest} and \ref{lem:modeltest}. Let $\delta > 0$ and $\e$ be as above. 
Then there exists an  algorithm running in time 
$O(n^4 \log^2 n \cdot \exp(1/\e^K) \cdot \log^2(1/\delta))$
which outputs with probability at least $1-\delta$ a subspace $V$ of codimension at most $\eps^{-C}$ as well as
a linear linear map $x \mapsto Tx$ and $c_1,c_2 \in \F_2^n$
satisfying $\Ex{x \in V+c_1}{  \widehat{f_x}^2(Tx+Tc_1+c_2
)}\geq \eps^{C}$.
\end{lemma}
Throughout the argument that follows, we shall assume that we have already chosen good parameters for
\bsgtest and \ModelTest so that the conclusions of Lemmas \ref{lem:bsgtest} and \ref{lem:modeltest}
hold. We also assume we have access to a good function $\phi$ as given by Lemma \ref{lem:sample-phi}.

To find the subspace $V$ we will apply Bogolyubov's lemma to the set identified by the procedure
\ModelTest. We shall look at the second half of the tuples in this subspace (coordinates $n+1$ to
$2n$) to find a linear choice function.

Let $h:\F_2^n \ra \{0,1\}$ be the (random) function defined by $h(y)=1$ if
$\ModelTest(u,(y,\phi(y)),\Gamma,c)=1$ and 0 otherwise. The error parameter $\delta'$ for \ModelTest
is taken to be $\delta/n^3$.
We shall apply the algorithm \bogolyubov 
from Lemma \ref{lem:bogolyubov} with queries to $h$ and with error parameter $\delta_1 = \delta/20$.

Note that the function $h$ is defined on points in $\F_2^n$.  Let $A^{(1)}$ and $A^{(2)}$ denote
projection on the first $n$ coordinates of the sets $A_{\phi}'^{(1)}$ and $A_{\phi}'^{(2)}$ given by
Lemma \ref{lem:modeltest}. 

Since the last $n$ coordinates are a function (namely $\phi$) of the first $n$ coordinates, we also have 
$|A_{\phi}'^{(1)}|\geq \theta N$, for $\theta$ a function of $\eps$ as defined in Claim
\ref{clm:subsetsize}. Also, with probability $1-\delta'$ for each input $x$, the inequality 
$1_{A^{(1)}}(x) \leq h(x) \leq 1_{A^{(2)}}(x)$ holds.

By Claim \ref{clm:noisy-bogolyubov},  we obtain a subspace $V_0$  of codimension $\theta^{-3}$ 
such that with probability at least 
$1-\delta_1 - \delta' \cdot n^2 \log n \cdot \poly(1/\theta,\log(1/\delta_1))  > 1-\delta/10$ , 
we have $V_0 \subseteq 4A^{(2)}$. Thus, each element $x \in V_0$ can we written as 
$x_1 + x_2 + x_3 + x_4$ for $x_1,x_2,x_3,x_4 \in A^{(2)}$. We next show that the set
\[ Z_0 \defeq
\inbraces{(x_1 + x_2 + x_3 + x_4, \phi(x_1) + \phi(x_2) + \phi(x_3) + \phi(x_4)) 
~\left\lvert~
\begin{array}{c}
x_1 + x_2 + x_3 +x_4 \in V_0, \\
x_1,x_2,x_3,x_4 \in A^{(2)}
\end{array}
\right.
}\]
is also a subspace of $\F_2^{2n}$. Observe that the value of $\phi(x_1) + \phi(x_2) + \phi(x_3) + \phi(x_4)$ is
uniquely determined by $x_1 + x_2 + x_3 + x_4$.
\begin{claim}\label{clm:zeta-linear}
There exists a linear map $\zeta: V_0 \to \F_2^n$ satisfying for any $x_1,x_2,x_3,x_4 \in A^{(2)}$ such
that $x_1+x_2+x_3+x_4 \in V_0$, we have 
$\phi(x_1)+\phi(x_2)+\phi(x_3)+\phi(x_4) = \zeta(x_1+x_2+x_3+x_4)$. Thus, the set $Z_0$ can be written
as   $Z_0 = \inbraces{(x,\zeta(x)) \suchthat x \in V_0}$ and is a subspace of $\F_2^n$.
\end{claim}
\begin{proof}
We first show that the value of $\phi(x_1)+\phi(x_2)+\phi(x_3)+\phi(x_4)$ is uniquely determined by 
$x_1+x_2+x_3+x_4$. By Lemma \ref{lem:modeltest}, we know that $\phi$ is a Freiman 8-homomorphism on
$A^{(2)}$ and hence it is also a Freiman 4-homomorphism. In particular, if for $x_1,x_2,x_3,x_4 \in
A^{(2)}$  and $x_1',x_2',x_3',x_4' \in A^{(2)}$, we have that 
$x_1+x_2+x_3+x_4 = x_1'+x_2'+x_3'+x_4'$, then it also holds that 
$\phi(x_1)+\phi(x_2)+\phi(x_3)+\phi(x_4) = \phi(x_1')+\phi(x_2')+\phi(x_3')+\phi(x_4')$. Thus, we
can write the set $Z_0$ as $\{ (x,\zeta(x)) \suchthat x \in V_0\}$, where $\zeta$ if some
function on $V$. We next show that $\zeta$ must be a linear function.

We first show that $\zeta(0) = 0$. Since $0 \in V_0$, we must have elements $x_1,x_2,x_3,x_4 \in A^{(2)}$
with the property that $x_1+x_2+x_3+x_4=0$, in other words, $x_1+x_2=x_3+x_4$. But since $\phi$ is also a Freiman 
2-homomorphism, we get that $\phi(x_1)+\phi(x_2) =\phi(x_3)+\phi(x_4)$, which implies that
$\phi(x_1)+\phi(x_2)+\phi(x_3)+\phi(x_4) = \zeta(0) = 0$.

Since $\phi$ is a Freiman 8-homomorphism on $A^{(2)}$ and $V_0 \subseteq 4A^{(2)}$, it follows
that $\zeta$ is a Freiman 2-homomorphism on $V_0$. Since $V_0$ is closed under addition, for 
$x,y \in V_0$ we can write $x + y = 0 + (x+y)$ with all four summands in $V_0$. Since $\zeta$ is
2-homomorphic, we get that $\zeta(x) + \zeta(y) = \zeta(0) + \zeta(x+y) = \zeta(x+y)$.
\end{proof}

We would like to use the linear map $\zeta$ to obtain the choice function on a coset of the space
$V_0$. However, the problem is that we do not \emph{know} the function $\zeta$. We get around this obstacle 
by generating random tuples $(x_1+x_2+x_3+x_4,\phi(x_1)+\phi(x_2)+\phi(x_3)+\phi(x_4))$ such that
$x_1+x_2+x_3+x_4$ and each $x_i \in A^{(2)}$. We show that for sufficiently many samples, the
sampled points span a large subspace $V$ of $V_0$. Since $\phi(x_1)+\phi(x_2)+\phi(x_3)+\phi(x_4) =
\zeta(x_1+x_2+x_3+x_4)$ on $V_0$, we will be able to obtain the desired linear map on the subspace $V$.

We sample a point as follows. For the $j^{th}$ sample, we generate four pairs 
$(x_1^{j}, \phi(x_1^j)), \ldots, (x_4^{j}, \phi(x_4^j))$. We accept the sample if all four pairs are
accepted by \ModelTest and if $x_1^j + x_2^j + x_3^j + x_4^j \in V$. If a sample is accepted, we
store the point $y^{j} = x_1^j + x_2^j + x_3^j + x_4^j$ and 
$\zeta(y^{j}) = \phi(x_1^j) + \phi(x_2^j) + \phi(x_3^j) + \phi(x_4^j)$.

Note that membership in $V_0$ can be tested efficiently since we know the basis for
$\ortho{V}_0$. We first estimate the probability that a point $(y,\zeta(y))$ for $y \in V_0$ is accepted
by the above test. This also gives a bound on the number of samples to be tried so that at least 
$t=O(n^2)$ samples are accepted.

\begin{claim}\label{clm:acceptance-subspace}
For a $y \in V_0$, the probability that a sample is accepted by the above procedure 
and the stored pair is equal to $(y,\zeta(y))$ is at least $\theta^4/4N$. Moreover, for some sufficiently large constant $C$, the probability that out of $C \exp(1/\theta^3) \cdot (1/\theta^4) \cdot t
\cdot \log(10/\delta)$ samples fewer than $t$ are accepted is at most $\delta/10$.
\end{claim}
\begin{proof}
Since the function $h(x) = 1$ exactly when \ModelTest accepts $(x,\phi(x))$, the probability
that a sample $(x_1,\phi(x_1)), \ldots, (x_4,\phi(x_4))$ is accepted and that $x_1+x_2+x_3+x_4=y$,
is equal to
\[
\prob{\bigwedge_{i=1}^4 (h(x_i) = 1) \wedge (x_1+x_2+x_3+x_4=y)}
~=~ (1/N) \cdot \Ex{h,x_1+x_2+x_3+x_4=y}{h(x_1) h(x_2) h(x_3) h(x_4)}
\]
As in Claim \ref{clm:noisy-bogolyubov}, we define the function 
$h' = \max\{1_{A^{(1)}}, \min\{h,1_{A^{(2)}}\}\}$. As before, we have that for each $x$, 
$\prob{h(x) \neq h'(x)} \leq \delta'$, and that $h'*h'*h'*h'(x) > \theta^4/2$ for 
each $x \in V_0$. We can now estimate the above expectation as
\begin{align*}
&~ \Ex{h,x_1+x_2+x_3+x_=y}{h(x_1) h(x_2) h(x_3) h(x_4)} \\
&~\geq~ \Prob{h,x_1+x_2+x_3+x_4=y}{\wedge_{i=1}^4 (h(x_i) = h'(x_i))} 
\cdot \Ex{h,x_1,x_2,x_3}{h'(x_1) h'(x_2) h'(x_3) h'(y+x_1+x_2+x_3)} \\
&~\geq~ (1-4\delta') \cdot h'*h'*h'*h'(y) \\
&~\geq~ (1-4\delta')\cdot (\theta^4/2) ~\geq~ \theta^4/4.
\end{align*}
The last inequality exploited the fact that $h'*h'*h'*h'(y) \geq \theta^4/2$ for $y \in V_0$. 

The probability that a sample is accepted is equal to the probability that one selects a 
pair $(y,\zeta(y))$ for \emph{some} $y \in V_0$. This is  
least $(|V_0|/N) \cdot (\theta^4/2) = \exp(-1/\theta^{3}) \cdot (\theta^4/2)$. 
The bound on the probability of accepting fewer than $t$ samples is then given by a Hoeffding bound.
\end{proof}

Let $(y^1,\zeta(y^1)), \ldots, (y^t,\zeta(y^t))$ be $t$ stored points corresponding to $t$ samples
accepted by the above procedure. 
The following claim analogous to Claim \ref{clm:large-span} shows that for $t=O(n^2)$, the 
projection on the first $n$ coordinates of these points must span a large subspace of $V_0$.

\begin{claim}\label{clm:loc-large-span}
Let $(y^1,\zeta(y^1)), \ldots, (y^t,\zeta(y^t))$ be $t$ points stored according to the above
procedure. For $t = n^2 + \log(10/\delta)$, the probability that 
$\cod(< y^1,\ldots,y^t >) \geq \cod(V_0) + \log(4/\theta^4)$ is at most $\delta/10$.
\end{claim}
\begin{proof}
Let $k = \cod(V_0) + 4\log(4/\theta)$ and let $S$ be any subspace of codimension $k$.
The probability that a sample $(x_1,\phi(x_1)), \ldots, (x_4,\phi(x_4))$ is  accepted and has 
$x_1+x_2+x_3+x_4 = y$ for a specific $y \in S$ is at most $1/N$. Thus, the probability that
an accepted sample $(y^j,\zeta(y^j))$ has $y^j \in S$, conditioned on being accepted, is at most
$(|S|/N)/((|V_0|/N) \cdot (\theta^4/2))$. Thus, the probability that all $t$ stored points lie
in \emph{any} subspace of co-dimension $k$ is at most
\[ \inparen{\frac{|S|/N}{(|V_0|/N) \cdot (\theta^4/2)}}^t \cdot \#\{\text{suspaces of
  co-dimension}~k\}
~=~ \inparen{\frac{\theta^4/4}{\theta^4/2}}^t \cdot 2^{n(n-k)}
~\leq~ 2^{-t} \cdot 2^{n^2},
\]
which is at most $\delta/10$ for $t = n^2 + \log(10/\delta)$.
\end{proof}

Let $V = <y^1,\ldots,y^t>$. The above claim shows that with high probability, the codimension of $V$ satisfies
$\cod(V) = \exp(1/\theta^3)$.
From the way the samples were generated, we also know $\zeta(y^1),
\ldots, \zeta(y^t)$. Since $\zeta$ is a linear function by Claim \ref{clm:zeta-linear}, we can extend
it to a linear transform $x \mapsto Tx$ such that $\forall x \in V$, $Tx = \zeta(x)$ (as in Section
\ref{sec:correlations}). 

We now show that there is a coset of $V$ on which $Tx$ identifies large Fourier coefficients of
the derivative $f_x$.  We define the set $Z \defeq \inbraces{(x,Tx) \suchthat x \in V}$. We will
find a coset of $Z$ such that a significant fraction of points in this coset are of
the form $(x,\phi(x)) \in A_{\phi}'^{(2)}$. Recall that a point $(x,\phi(x))$ in 
$A_\phi'^{(2)}$ satisfies $|\fxhat(\phi(x))| \geq \gamma = O(\e^{16})$. Thus, $Tx$ will be a
linear function selecting large Fourier coefficients for a significant fraction of points in this coset.

The following claim shows the existence of such a coset.

\begin{claim}\label{clm:good-coset}
The sets $Z+A_{\phi}'^{(1)}$ and $Z+A_{\phi}'^{(2)}$ both consist of at most $(1/\theta) \cdot
(N/\abs{Z})$  cosets of $Z$. Hence, for some $c \in A_{\phi}'^{(1)}$ we have
$|(Z+c) \cap A_{\phi}'^{(2)}| \geq |(Z+c) \cap A_{\phi}'^{(1)}| \geq \theta^{2} \cdot \abs{Z}$.
\end{claim}
\begin{proof}
Since $Z \subseteq 4 A_{\phi}'^{(2)}$ and $A_{\phi}'^{(1)} \subseteq A_{\phi}'^{(2)}$, we have that
\[Z + A_{\phi}'^{(1)} ~\subseteq~ Z + A_{\phi}'^{(2)} ~\subseteq~ 5 A_{\phi}'^{(2)} ~\subseteq~ 5
A_{\phi}^{(2)} .\]  
The last inclusion follows from
the fact that $A_{\phi}'^{(2)}$ was obtained by intersecting $A_{\phi}'^{(2)}$ (given by Lemma 
\ref{lem:bsgtest}) with a subspace.

We know from Lemma \ref{lem:bsgtest} that 
$\smallabs{A_{\phi}^{(2)} + A_{\phi}^{(2)}} \leq (2/\rho)^8 \cdot N \leq (2/\rho)^8 \cdot (6/\rho)
\cdot \smallabs{A_{\phi}^{(2)}}$. Lemma \ref{lem:plunnecke} (\Plunnecke's inequality) then gives
that $\smallabs{5 A_{\phi}^{(2)}} \leq (6/\rho)^{45} \cdot \smallabs{A_{\phi}^{(2)}} \leq (1/\theta)
\cdot \smallabs{A_{\phi}^{(2)}} \leq (1/\theta) \cdot N$. Thus, 
$\smallabs{Z + A_{\phi}'^{(2)}} \leq (1/\theta) \cdot N$ and it is the union of at most
$(1/\theta) \cdot (N/|Z|)$ cosets.

Since $A_{\phi}'^{(1)} \subseteq Z + A_{\phi}'^{(1)}$, there must exist at least one coset $Z+c$ for
$c \in A_{\phi}'^{(1)}$, such that 
\[\abs{(Z+c) \cap A_{\phi}'^{(1)}} ~\geq~ \frac{\smallabs{A_{\phi}'^{(1)}}}{(1/\theta) \cdot
  (N/|Z|)} ~\geq~ \theta^2 \cdot |Z|,\]
where the last inequality used the fact that $\smallabs{A_{\phi}'^{(1)}} \geq \theta N$, as
guaranteed by Lemma \ref{lem:modeltest}.
\end{proof}

We now show how to computationally identify this coset of $Z$. We will simply sample a sufficiently large number of points on
which \ModelTest answers 1. We will then divide the points into different cosets of $Z$ and pick the coset with the most number of elements. The following claim shows that this procedure succeeds in
finding the desired coset with high probability.

\begin{claim}\label{clm:find-coset}
Let $s =  C \cdot (N/|Z|) \cdot (\log(1/\delta)/\theta^5) \leq 
C \cdot \exp(1/\theta^3) \cdot (\log(1/\delta)/\theta^5)$ for a sufficiently large constant $C$. 
There exists an algorithm which runs in time $O(n^3 \cdot s^2)$ and finds, with
probability at least $1-\delta/5$, a point $c \in A_{\phi}'^{(2)}$ such that 
$\smallabs{(Z+c) \cap A_{\phi}'^{(2)}} \geq (\theta^2/2) \cdot |Z|$.
\end{claim}
\begin{proof}
We sample $s$ independent elements of the form $(x,\phi(x))$ and reject all the ones on which
\ModelTest outputs 0, where we run \ModelTest with error parameter $\delta' = \delta/(10s)$.
For some $r \leq s$, let $(x_1,\phi(x_1)), \ldots, (x_r,\phi(x_r))$ be the
accepted elements. 

For each $i,j \leq r$, we test if $(x_i, \phi(x_i))$ and $(x_j, \phi(x_j))$ lie in the same coset of
$Z$, by checking if $(x_i-x_j, \phi(x_i) - \phi(x_j)) \in Z$. This takes time $O(n^3)$ for each
$i,j$ as we need to check if $(x_i-x_j, \phi(x_i) - \phi(x_j))$ can be expressed as a linear
combination of the basis vectors for $Z$, which  requires solving a system of linear equations.

Lying in the same coset is an equivalence relation, which divides the points $(x_1,\phi(x_1)), \ldots,
(x_r,\phi(x_r))$ into equivalence classes. We pick the class with the maximum number of
elements. Since $(0,0) \in Z$, for any element $(x_i,\phi(x_i))$ in this class,  we can write the
coset as $Z + (x_i,\phi(x_i))$. We thus pick an arbitrary element of the form $(x_i,\phi(x_i))$
in the largest class and output $c = (x_i,\phi(x_i))$.

The running time of the above algorithm is $O(s^2 \cdot n^3)$. We need to argue that with
probability at least $1-\delta/5$, the coset $Z+c$ with the maximum number of samples satisfies
$|(Z+c) \cap A_{\phi}'^{(2)}| \geq (\theta^2/2) \cdot |Z|$. 

With probability at least $1 - \delta' \cdot s = 1 - \delta/10$, \ModelTest answers 1
on all elements in  $A_{\phi}'^{(1)}$ and 0 on all elements outside $A_{\phi}'^{(2)}$. For any coset
of the form $Z+c$, let $N(Z+c)$ be the number of samples that land in the coset.
Conditioned on the correctness of \ModelTest, we have that for any coset of the form $Z+c$,
\[
 s \cdot \frac{\smallabs{(Z+c) \cap A_{\phi}'^{(1)}}}{N} ~\leq~ 
\ex{N(Z+c)} ~\leq~  
s \cdot \frac{\smallabs{(Z+c) \cap A_{\phi}'^{(2)}}}{N},\]
which by definition of $s$ implies that
\[C \cdot \frac{\log(1/\delta)}{\theta^5} \cdot \frac{\smallabs{(Z+c) \cap A_{\phi}'^{(1)}}}{|Z|} ~\leq~ 
\ex{N(Z+c)} ~\leq~  
C \cdot \frac{\log(1/\delta)}{\theta^5} \cdot \frac{\smallabs{(Z+c) \cap A_{\phi}'^{(2)}}}{|Z|} .\]

By a Hoeffding bound, the probability that $N(Z+c)$ deviates by an additive 
$(C/4) \cdot (\log(1/\delta)/\theta^3)$ from the expectation is
at most $\delta \cdot \exp( - C' (1/\theta^3))$ for any fixed coset. Since the number of cosets is at most $(1/\theta)
\cdot \exp(1/\theta^3)$ by Claim \ref{clm:good-coset}, the probability that on \emph{any} coset $N(Z+c)$
deviates from the expectation by the above amount is at most 
$\delta \cdot \exp( - C' (1/\theta^3)) \cdot (1/\theta) \cdot \exp(1/\theta^3) < \delta/10$ for an
appropriate value of $C'$.

By Claim \ref{clm:good-coset}, we know that there is a coset $Z+c$ with 
$|(Z+c) \cap   A_{\phi}'^{(1)}| \geq \theta^2|Z|$ and hence 
$\ex{N(Z+c)} \geq C \cdot (\log(1/\delta)/\theta^3)$. By the above deviation bound, we should have that $N(Z+c) \geq (3C/4) \cdot (\log(1/\delta)/\theta^3)$ for this coset.
Thus, the coset with the maximum number of samples, say $Z+c'$, will certainly also satisfy
$N(Z+c') \geq (3C/4) \cdot (\log(1/\delta)/\theta^3)$. Again, by the deviation bound, it must be true that $\ex{N(Z+c')} \geq (C/2) \cdot (\log(1/\delta)/\theta^3)$,
and hence $|(Z+c) \cap   A_{\phi}'^{(2)}| \geq \theta^2|Z|/2$.
\end{proof}

We can now combine the previous argument to prove Lemma \ref{lem:lin-locchoice}.

\begin{proofof}{of Lemma \ref{lem:lin-locchoice}}
We follow the steps described above to find the subspace $V_0$, and subsequently the subspace $V$ together with the
transformation $T$. This immediately yields the subspace $Z = \{(x,Tx) \suchthat x \in V\}$. 
Claim \ref{clm:find-coset} finds $c = (c_1,c_2) \in \F_2^n$  such that a fraction of at least
$\theta^2/2$ of points $(y+c_1,Ty+c_2)$ in the coset $Z+(c_1,c_2)$
are of the form $(x,\phi(x))$ for $(x,\phi(x)) \in  A_{\phi}'^{(2)}$, and so
$|\fxhat^2(\phi(x))| \geq \gamma = O(\e^{16})$. Since $(y,Ty+c_2) = (x+c_1,\phi(x))$ for these
points, we have $T(x+c_1) + c_2 = \phi(x)$. This implies
\begin{equation}\label{ave}
\Ex{x \in c_1+V}{  \widehat{f_x}^2(Tx+Tc_1 + c_2)} ~\geq~ (\theta^2/2) \cdot \gamma^2 ~\geq~ \e^C.
\end{equation}
The errors in the application of Bogolyubov's lemma and in Claims \ref{clm:acceptance-subspace},
\ref{clm:loc-large-span} and \ref{clm:find-coset} add up to $\delta/2 < \delta$. The running time is
dominated by the $C \exp(1/\theta^3) \cdot (1/\theta^4) \cdot t \cdot \log(10/\delta)$ 
calls to \ModelTest in Claim \ref{clm:acceptance-subspace} for $t = O(n^2)$. Since each call to
\ModelTest takes $O(n^2 \log n \cdot \poly(1/\e) \cdot \log(\delta/n^3) )$ time, the total running
time is $O(n^4 \log^2 n \cdot \exp(O(1/\theta^3)) \cdot \log^2(1/\delta))$.
\end{proofof}

\subsubsection*{Fourier analysis over a subspace} 
To begin with we collect some basic facts about Fourier analysis over a
subspace of $\F_2^n$, which will be required for the remaining part of the argument.
Let $f: \F_2^n \to \R$ be a function and let $W
\subseteq \F_2^n$ be a subspace. We define the Fourier coefficients of $f$ with respect to the
subspace as the correlation with a linear phase over the subspace. 

As in the case of Fourier analysis over $\F_2^n$, it is easy to verify that the functions 
$\inbraces{\chi_{\alpha}}_{\alpha  \in W}$ with $\chi_{\alpha}(x) \defeq (-1)^{\langle\alpha,
  x\rangle}$ form an orthonormal basis for functions from $W$ to $\R$
with respect to the inner product $\ip{f_1,f_2}_W ~\defeq~ \Ex{x \in W}{f_1(x) f_2(x)}$. Thus the
dual group $\hat{W}$ of these basis functions is isomorphic to $W$. As in the case of $\F_2^n$, we
have Parseval's identity saying that $\sum_{\alpha \in W} \ip{f,\chi_{\alpha}}_W^2 = \Ex{x\in
  W}{f^2(x)}$.

It is easy to modify the proof of the Goldreich-Levin theorem so that it can be used to identify the
linear functions $\chi_{\alpha}$ for $\alpha \in W$ that have large correlation with a Boolean
function $f$ over a subspace $W$. We omit the details.

\begin{theorem}[Goldreich-Levin theorem for a subspace]\label{thm:gl-subspace}
Let $\gamma,\delta>0$ and $W \subseteq \F_2^n$ be a given subspace.
There is a randomized algorithm which, given oracle access to a function $f: \F_2^n \to \pmone$,
runs in time $O(n^2\log n \cdot \poly(1/\gamma, \log(1/\delta)))$ and 
outputs a list $L = \{\alpha_1,\ldots,\alpha_k\}$ with each $\alpha_i \in W$ such that 
\begin{itemize}
\item $k = O(1/\gamma^2)$.
\item $\prob{\exists \alpha_i \in L~ \smallabs{\ip{f,\chi_{\alpha_i}}_W} \leq \gamma/2} \leq \delta$.
\item $\prob{\exists \alpha \notin L~ \smallabs{\ip{f,\chi_{\alpha_i}}_W} \geq \gamma} \leq \delta$.
\end{itemize}
\end{theorem}

\subsection{Finding a quadratic phase on a subspace} \label{sec:local-quadratic}

In order to deduce the refined inverse theorem (Theorem \ref{thm:localinverse}) for $p=2$, we need
to redo the symmetry argument and integration phase with this local expression obtained in Lemma \ref{lem:lin-locchoice}. The modifications to
Samorodnitsky's approach are relatively minor but we give complete proofs nonetheless. One
significant difference is that we will need to take Fourier transforms relative to subspaces.

We begin by obtaining a subspace $W \leqslant V$ on which the matrix $T$ obtained in the previous step is symmetric, thereby providing the ``local" analogue of Lemma \ref{lem:symmetrization}.


\begin{lemma}[Symmetry Argument]\label{lem:findsymmetric}
Given a subspace $V$ and a linear map $T$ with the property that
\[\E_{x \in c_1+ V}  \widehat{f_x}^2(Tx+z_c)\geq \eps^{C},\]
we can output a subspace $W \leqslant V$ of codimension at most $\log(\eps^{-C})$ inside $V$ together with a symmetric matrix $B$ on $W$ with zero diagonal such that
\[\E_{x \in c_1+ W}\widehat{f_x}^2(Bx+z_c)\geq\eps^{C}\]
in time $O(n^3)$.
\end{lemma}

\begin{proof}
We let $g(x)=(-1)^{\langle x,Tx+z_c\rangle}$ and $F(x)=\widehat{f_x}^2(Tx+z_c)$, and begin by noting that by Lemma 6.11 in \cite{Samorodnitsky07}, we have that $g(x)=-1$ implies $F(x)=0$. 
Therefore we have
\[\eps^C \leq \E_{x \in c_1+ V}  \widehat{f_x}^2(Tx+z_c)=\E_{x \in c_1+V} g(x)F(x)=\E_{x \in V} g^{c_1}(x)F^{c_1}(x),\]
we have written $h^{y}(x)$ for the shift $h(x+y)$. Taking the Fourier transform relative to the subspace $V$, we obtain
\[\eps^C \leq (\sum_{\alpha \in \widehat{V}} \widehat{g^{c_1}}(\alpha)\widehat{F^{c_1}}(\alpha))^2,\]
and by the Cauchy-Schwarz inequality and Parseval's theorem this is bounded above by 
\[\sum_{\alpha \in \widehat{V}} \widehat{g^{c_1}}(\alpha)^2 \sum_{\alpha\in \widehat{V}} \widehat{F^{c_1}}(\alpha)^2\leq \E_{x \in V} g^{c_1}*_V g^{c_1}(x).\]
The latter (local) convolution can easily be computed:
\[ g^{c_1}*_V g^{c_1}(x)=\E_{y \in V} (-1)^{ \langle x+y+c_1, T(x+y)+c_2\rangle} (-1)^{\langle y+c_1, Ty+c_2\rangle} = g^{c_1}(x) (-1)^{\langle c_1,c_2 \rangle} \E_{y \in V} (-1)^{\langle (T+T^T) x, y \rangle}.\]
The final expectation gives the indicator function of the subspace 
\[W'=\{x \in V : \langle(T+T^T)x,y \rangle=0 \mbox{ for all } y \in V\},\]
that is, $W'$ is a linear subspace on which $T$ is symmetric. Note that $W'$ is the space of solutions of a linear system of equations, a basis of which can be computed by Gaussian elimination in time $O(n^3)$. 

We denote the map that takes $x$ to $Tx$ for $x \in W'$ by $B$. We have just shown that  
\[|\E_{x \in V} 1_{W'}(x) g^{c_1}(x)|\geq\eps^{C},\]
and in particular since $g$ is bounded, we quickly observe that $W'$ has density at least $\eps^{C}$ inside $V$. This means the codimension can have gone up by at most $\log(\eps^{-C})$, which is negligible in the grand scheme of things. 

It remains to ensure that $B$ has zero diagonal. Again this can be rectified in a small number of steps. Denote this diagonal by $v \in \F_2^n$. Let $W=W' \cap< v+z_c>^{\perp}$ if $\langle c_1,c_2\rangle=0$, otherwise intersect $W'$ with the (unique) coset of $< v+z_c>^{\perp}$. Since $\langle x,Bx\rangle=\langle x,v\rangle$ over $\F_{2}$, we have that $\langle x+c_1,v+z_c\rangle = \langle x, Bx +z_c \rangle +\langle c_1 ,c_2 \rangle$, and thus by Lemma 6.11 in \cite{Samorodnitsky07} if $x+c_1 \in W'$ but $\notin W$, that is, $x+c_1 \notin <v+z_c>^\perp$, then $\widehat{f_x}^2(Bx+z_c)=0$.

Hence we obtain
\[2 \E_{x \in c_1+ W}\widehat{f_x}^2(Bx+z_c) =\E_{x \in c_1+ W'} \widehat{f_x}^2(Bx+z_c),\]
which yields the desired conclusion.
\end{proof}

Finally, we need to perform the integration. The procedure is very similar to Lemma \ref{lem:integration}, but again we have to work relative to a subspace. 


\begin{lemma}[Integration Step]\label{lem:FCtoCorr2}
Let $f:\F_{2}^{n} \ra [-1,1]$. Let $B$ be a symmetric $n \times n$ matrix with zero diagonal such
that $\E_{x \in c_1 + W} \widehat{f_x}^2(Bx+z_c) \geq \eps^C$. Let $A \in \F_2^{n \times n}$ be a
matrix such that $B = A + A^T$.
Then there exist, for every $y \in \F_{2}^{n}$, a vector $r_{y} \in W$ 
such that
\[\E_{y \in \comp{W}}|\E_{x \in y+W} f(x)(-1)^{\langle x,Ax\rangle + \ip{By,x}+ \langle r_y, x\rangle}|
\geq \eps^C.\]
\end{lemma}
\begin{proof}
Consider the quadratic phase $g(x)=(-1)^{\langle x,Ax\rangle}$ and the linear phase $l(z)=(-1)^{\langle z,z_c\rangle}$. (Note that this is where we require $B$ to have zero diagonal.) We shall first prove that  
 \[\E_{x \in c_1+ W}\widehat{f_x}^2(Bx+z_c)= \E_{x \in c_1+ W}(\E_{y \in \comp{W}}\langle f_x,g_x l \rangle_{y+W})^{2} \leq \E_{y \in \comp{W}}\sum_{\alpha\in \widehat{W}} \widehat{(fgl)^{y}}^{2}(\alpha)\widehat{(fg)^{y}}^{2}(\alpha),\] 
where again we have written $h^{y}(x)$ for the shift $h(x+y)$ and the final Fourier transform is taken with respect to $W$. The equality follows from the fact that 
\[\widehat{f_x}(Bx+z_c)=\E_yf_x(y)(-1)^{\langle y,Bx+z_c\rangle}=\E_{y\in \comp{W}} \E_{z \in y+W}f_x(z)(-1)^{\langle z, Bx+z_c\rangle}\]
and so
\[(-1)^{\langle x,Ax\rangle}\widehat{f_x}(Bx+z_c)=\E_{y\in \comp{W}} \E_{z\in y+W}f_x(z)(-1)^{\langle z+x,A(z+x)\rangle +\langle z,Az\rangle}l(z)=\E_{y\in \comp{W}}\langle f_x,g_x l\rangle_{y+W}, \]
where the inner product is taken over the translate $y+W$.
For the inequality write
\[ \E_{x \in c_1+ W}(\E_{y \in \comp{W}}\langle f_x,g_x l\rangle_{y+W})^{2} \leq \E_{y\in \comp{W}} \E_{x \in c_1+ W}\langle f_x,g_xl\rangle_{y+W}^{2},\]
which equals
\[\E_{y\in \comp{W}}\E_{x \in c_1+ W} (\E_{z \in y+W} fgl(z)fg(z+x))^2=\E_{y\in \comp{W}}\E_{x \in W} (\E_{z \in y+W} fgl(z)fg(z+x+c_1))^2,\]
which in turn can be reexpressed as
\[\E_{y\in \comp{W}}\E_{x \in W} (\E_{z \in W} (fgl)^{y}(z)(fg)^{y}(z+x+c_1))^2 =\E_{y\in \comp{W}}\E_{x \in W}  ((fgl)^{y}*_{W}(fg)^{y})(x+c_1)^2.\]
Taking the Fourier transform with respect to $W$, it can be seen that the latter expression equals
\[\E_{y\in \comp{W}}\sum_{\alpha \in \widehat{W}} \widehat{(fgl)^{y}}^{2}(\alpha)\widehat{(fg)^{y}}^{2}(\alpha),\]
completing the proof of the claim from the beginning. But since all functions involved are bounded,
\[\E_{y\in \comp{W}}\sum_{\alpha \in \widehat{W}} \widehat{(fgl)^{y}}^{2}(\alpha)\widehat{(fg)^{y}}^{2}(\alpha)\leq  \E_{y\in \comp{W}} \sup_{\alpha\in \widehat{W}}|\widehat{(fg)^{y}}(\alpha)|.\]
Now for each $y \in \comp{W}$, we fix a $\alpha_{y} \in \widehat{W}$ such that the supremum is attained. Then we have shown that 
\[\eps^C\leq \E_{y\in \comp{W}}|\widehat{(fg)^{y}}(\alpha_{y})| = \E_{y\in \comp{W}}| \E_{x \in W} f(x+y)(-1)^{\langle x+y,A(x+y)\rangle+\langle \alpha_{y},x\rangle}|,\]
which, after some rearranging of the phase, completes the proof.
\end{proof}

\subsection{Obtaining a quadratic average} \label{sec:quadratic-average}
Finally, we use the subspace $W$ from Section \ref{sec:local-quadratic} to obtain the required
quadratic average.

\begin{lemma}\label{lem:find-linear-parts}
Let $W \leqslant \F_2^n$ be a subspace with $\cod(V) \leq (1/\e^C)$. Let $A \in \F_2^{n \times
  n}$ and $B=A+A^T$ be such that there exist vectors $r_y \in W$ for each $y \in
\comp{W}$ satisfying
\[\Ex{y \in \comp{W}} {\abs{\Ex{x \in y+W}{f(x)(-1)^{\ip{x,Ax} +\ip{By,x} + \ip{r_y, x}}}}} \geq
\sigma.\]
Then for $\delta > 0$, one can find in time
 $n^2 \log n \cdot |\comp{W}| \cdot \poly(1/\sigma, \log(1/\delta))$
a quadratic average with a vector $l_y$ and a constant $c_y$ for each $y \in \comp{W}$ satisfying
\[\Ex{y \in \comp{W}} {\Ex{x \in y+W}{f(x)(-1)^{\ip{x,Ax} +\ip{l_y, x} + c_y}}} \geq  \sigma^2/10.\]
\end{lemma}
\begin{proof}
Let $h_y(x) \defeq f(x) (-1)^{\ip{x,Ax} + \ip{x,By}}$. By assumption we immediately find that
\[\Ex{y \in \comp{W}} {\abs{\Ex{x \in y+W}{h_y(x)(-1)^{\ip{r_y, x}}}}} 
~=~  \Ex{y \in \comp{W}} {\abs{\Ex{x \in W}{h_y^y(x)(-1)^{\ip{r_y, x}}}}}
~\geq~ \sigma.\]
Here $h_y^y(x) = h_y(x+y)$ as before. Without loss of generality, we may assume that the vectors $r_y$ 
maximize the above expression.
Thus, we know that on average (over $y$),
the functions $h_y^y$ have a large Fourier coefficient (that is, significant correlation with some vector $r_y \in W$) 
over the subspace $W$. For every $y \in \comp{W}$, we will use Theorem \ref{thm:gl-subspace} to find this Fourier coefficient when it is indeed large. For those $y$ for which the expression $\abs{\Ex{x \in W}{h_y^y(x)(-1)^{\ip{r_y, x}}}}$
is small for all $r_y \in W$, we will simply pick an arbitrary phase.

Let us describe this procedure in more detail. First, by an averaging argument we know that 
\[
\Ex{y \in \comp{W}} {\abs{\Ex{x \in W}{h_y^y(x)(-1)^{\ip{r_y, x}}}}} ~\geq~ \sigma
~\implies~
\Prob{y \in \comp{W}} {\abs{\Ex{x \in W}{h_y^y(x)(-1)^{\ip{r_y, x}}}} \geq \sigma/2} ~\geq~ \sigma/2
.
\]
Let $S \defeq \{y \in \comp{W} \suchthat \abs{\Ex{x \in W}{h_y^y(x)(-1)^{\ip{r_y, x}}}} \geq \sigma/2\}$.
The above inequality shows that $|S| \geq (\sigma/2) \cdot \comp{W}$. Now for each $y \in \comp{W}$, we run the Goldreich-Levin algorithm for the subspace $W$ from Theorem
\ref{thm:gl-subspace} with the function $h_y^y$, the parameter $\gamma = \sigma/2$ and error
probability $\delta^2/2$. 

For each $y \in S$ the algorithm finds, with probability $1-\delta^2$, an $r_y' \in W$ and a $c_y
\in \F_2$ satisfying $\Ex{x \in W}{h_y^y(x)(-1)^{\ip{r_y', x} + c_y}} ~\geq~ \sigma/4$. Thus, with
probability $1-\delta/2$, it finds such an $r_y'$ for at least a $1-\delta$ fraction of $y \in S$.
For $y \notin S$, that is for those $y$ for which the algorithm fails to find a good linear phase, we choose an $r_y'$
arbitrarily. If we can force the contribution of terms for $y \notin S$ to be non-negative, then we have that
with probability $1-\delta/2$
\[ \Ex{y \in \comp{W}} {1_{S}(y) \cdot \Ex{x \in W}{h_y^y(x)(-1)^{\ip{r_y', x} + c_y}}} ~\geq~ 
(1-\delta) \cdot (\sigma/2) \cdot (\sigma/8) \geq \sigma^2/9 .
\]
It remains to choose constants $c_y$ for $y \notin S$ in such a way that their contribution to the average is
non-negative. Consider the two potential assignments $c_y = 0 ~\forall y \notin S$ and $c_y = 1 ~\forall y \notin
S$. Clearly the contribution of the terms
for $y \notin S$ must be non-negative for at least one of the aforementioned assignments, in which case we obtain%
\[ \Ex{y \in \comp{W}} {\Ex{x \in W}{h_y^y(x)(-1)^{\ip{r_y', x} + c_y}}} ~\geq~ \sigma^2/9 .
\]
In order to determine which of the two assignments works, we can try both sets of signs and estimate the corresponding quadratic average using $O((1/\sigma^4)
\cdot \log(1/\delta))$ samples, and choose the set of signs for which the estimate is larger. By
Lemma \ref{lem:hoeffding-sample}, with probability at least $1-\delta/2$, we select a set of values
$c_y$ such that 
\[ \Ex{y \in \comp{W}} {\Ex{x \in y+W}{f(x)(-1)^{\ip{x,Ax} + \ip{x,By} + \ip{x,r'_y} + c_y}}} ~=~
\Ex{y \in \comp{W}} {\Ex{x \in W}{h_y^y(x)(-1)^{\ip{r_y', x} + c_y}}} ~\geq~ \sigma^2/10 .
\]
Choosing $l_y = By + r_y'$ then completes the proof.
\end{proof}

\subsection{Putting things together} \label{sec:final-quadratic-average}
We now give the proof of Theorem \ref{thm:FindQuadraticAverage}.

\begin{proofof}{of Theorem \ref{thm:FindQuadraticAverage}}
For the procedure \FindQuadraticAverage the function $\phi(x)$ will be sampled using Lemma 
\ref{lem:sample-phi} as required. We start with a random $u = (x,\phi(x))$
and a random choice of the parameters $\gamma_1,\gamma_2,\gamma_3$ as
described in the analysis of \bsgtest. We also choose the map $\Gamma$ and the value $c$ randomly
for \ModelTest. 
 We run the algorithm in Lemma \ref{lem:lin-locchoice} 
using \bsgtest and \ModelTest  with the above parameters, and with error parameter $1/4$.

Given a coset of the  subspace $V$ and the map $T$, we find a subspace $W \subseteq V$ and a
symmetric matrix $B$ with zero diagonal, using Lemma \ref{lem:findsymmetric}. We then use the
algorithm in Lemma \ref{lem:find-linear-parts} to obtain the required quadratic average, with
probability $1/4$.

Given a quadratic average $Q(x)$,
we estimate $\abs{\ip{f,Q}}$ using $O((1/\sigma^4) \cdot \log^2(\theta/\delta))$
samples. If the estimate is less than $\sigma^2/20$, we discard $Q$ and repeat
the entire process.  For a $M$ to be chosen later, if we do not find
a quadratic average in $M$ attempts, we stop and output $\bot$.

With probability $\rho/2$, all samples of $\phi(x)$ (sampled with error $1/n^5$)
correspond to a good function $\phi$. Conditioned on this, we have a good choice of
$u$ and $\gamma_1,\gamma_2,\gamma_3$ for \bsgtest with probability $\rho^3/24$.
Also, we have a good choice of the map $\Gamma$ and $c$ for \ModelTest with probability at least
$\theta/2 = \e^{O(1)}$.
Conditioned on the above, 
the algorithm in Lemma \ref{lem:lin-locchoice} finds a good transformation with probability
$3/4$ and thus the output of the algorithm in Lemma \ref{lem:find-linear-parts} is a good quadratic
average with probability at least $1/2$.

Thus, for $M = O((1/\rho^4 ) \cdot (1/\theta) \log(1/\delta))$, the algorithm stops in 
$M$ attempts with probability at least $1-\delta/2$. By choice of the number of
samples above, the probability that we estimate $\abs{\ip{f,(-1)^q}}$ incorrectly at any step is at most 
$\delta/2M$. Therefore we output a good quadratic average with probability at least $1-\delta$.

The complexity of the quadratic average obtained, which is equal to the co-dimension of the space
$W$, is at $O(1/\theta^3) = O(1/\e^C)$.
The running time of each of the $M$ steps is dominated by that of the algorithm in 
Lemma \ref{lem:lin-locchoice}, which is $O(n^4 \log^2 n \cdot \exp(1/\e^K))$.
We conclude that the total running time is $O(n^4 \log^2 n \cdot \exp(1/\e^K) \cdot \log(1/\delta))$.
\end{proofof}

\section{Discussion}
\label{sec:discussion}

One way in which one might want extend the results in this paper is to consider the cyclic group of integers modulo of prime $\Z_N$. A (linear) Goldreich-Levin algorithm exists in this context \cite{AkaviaGS03}, and some quadratic decomposition theorems have been proven (see for example \cite{GW4}). However, strong quantitative results involving the $U^3$ norm require a significant amount of effort to even state. 

For example, the role of the subspace relative to which the quadratic averages are defined will be played by so-called Bohr sets, which act as approximate subgroups in $\Z_N$. Moreover, it is no longer true that the inverse theorem can guarantee the existence of a globally defined quadratic phase with which the function correlates; instead, this correlation may be forced to be (and remain) local.

Since there is an informal dictionary for translating analytic arguments from $\F_p^n$ to $\Z_N$, it seems plausible that many of our arguments could be extended to this setting, at the cost of adding a significant layer of (largely technical) complexity to the current presentation.

\section{Acknowledgements}
 
The authors would like to thank Tim Gowers, Swastik Kopparty, Tom Sanders and Luca Trevisan for
helpful conversations.


\bibliographystyle{amsalpha}
\bibliography{macros,quadratic}


\end{document}